\documentclass[journal]{IEEEtran}
\usepackage{comment}
\usepackage{soul}
\soulregister{\cite}7 
\soulregister{\ref}7 

\usepackage[T1]{fontenc} 
\usepackage{color}
\usepackage{multirow}
\usepackage{array}
\usepackage{amsmath,amsthm,amsfonts,amssymb,bm,mathrsfs}
\usepackage{url}
\usepackage{soul} 
\usepackage{graphicx} 
\usepackage{stfloats}
\usepackage{enumitem}
\usepackage{algorithmic}
\usepackage{algorithm}
\usepackage{overpic}
\usepackage{cite} 
\usepackage{wasysym}
\usepackage[colorlinks,citecolor=black, filecolor=black, linkcolor=black,urlcolor=black]{hyperref}
\usepackage[table]{xcolor}
\definecolor{slg}{gray}{0.95}

\allowdisplaybreaks[4]

\makeatletter
\newcommand{\removelatexerror}{\let\@latex@error\@gobble}
\makeatother
\usepackage{booktabs, makecell, tabularx, threeparttable}
\usepackage{caption}
\newcolumntype{C}{>{\centering\arraybackslash}X} 

\usepackage{stfloats}
\usepackage{siunitx}
\usepackage{multirow}
\usepackage{subfigure}

\newtheorem{theorem}{Theorem}
\newtheorem{corollary}{Corollary}

\def\ie{{i.e.}}

\def\etal{\textit{et al.}~}
\def\FTX{\mathcal{F}_{\rm{TX}}}
\def\FRXN{{\mathcal{F}_{{\rm RX}{_n}}}}
\def\rxn{${\rm RX}_n$}
\def\tx{$\rm{TX}$}
\def\rxs{$\rm{RXs}$}

\usepackage{hyperref}

\begin{document}
\title{Self-Critical Alternate Learning based Semantic Broadcast Communication}
\author{Zhilin~Lu,~Rongpeng~Li,~Ming~Lei,~Chan~Wang, ~Zhifeng~Zhao,~and~Honggang~Zhang\vspace{-1cm}
    \thanks{
  Z.~Lu,~R. Li, M. Lei, and C. Wang are with Zhejiang University, China (email: \{lu\_zhilin, lirongpeng, lm1029, 0617464\}@zju.edu.cn);~Zhifeng Zhao and Honggang Zhang are with Zhejiang Lab as well as Zhejiang University, China (email: \{zhaozf,honggangzhang\}@zhejianglab.com). } 
}
	\maketitle

\begin{abstract}
Semantic communication (SemCom) has been deemed as a promising communication paradigm to break through the bottleneck of traditional communications. Nonetheless, most of the existing works focus more on point-to-point communication scenarios and its extension to multi-user scenarios is not that straightforward due to its cost-inefficiencies to directly scale the JSCC framework to the multi-user communication system. 
Meanwhile, previous methods optimize the system by differentiable bit-level supervision, easily leading to a ``semantic gap''.   
Therefore, we delve into multi-user broadcast communication (BC) based on the universal transformer (UT) and propose a reinforcement learning (RL) based self-critical alternate learning (SCAL) algorithm, named SemanticBC-SCAL, to capably adapt to the different BC channels from one transmitter (\tx) to multiple receivers (\rxs) for sentence generation task.  
In particular, to enable stable optimization via a non-differentiable semantic metric, we regard sentence similarity as a reward and formulate this learning process as an RL problem. Considering the huge decision space, we adopt a lightweight but efficient self-critical supervision to guide the learning process. Meanwhile, an alternate learning mechanism is developed to provide  
cost-effective learning, in which the encoder and decoders are updated asynchronously with different iterations. Notably, the incorporation of RL makes SemanticBC-SCAL compliant with any user-defined semantic similarity metric and simultaneously addresses the channel non-differentiability issue by alternate learning. Besides, the convergence of SemanticBC-SCAL is also theoretically established. Extensive simulation results have been conducted to verify the effectiveness and superiorness of our approach, especially in low SNRs.     
    
\end{abstract}

	
\IEEEpeerreviewmaketitle
	
\section{Introduction}

 With the rapid development of communication technologies and the rise of artificial intelligence (AI), tremendous intelligent applications, such as extended reality (XR), holographic communication and autonomous driving, flourish and impose stringent communication requirements 
 \cite{wang2023road}, which are beyond the capabilities of 5G. In order to deal with this performance dilemma, researchers are resorting to reap the merits of AI, and deep learning (DL) enabled semantic communication (SemCom) emerges as one of the promising solutions. In particular, targeted at the semantic and/or effectiveness level communications \cite{weaver1949recent}, SemCom has attained intense research interest \cite{lu2022semantics} with encouraging results in text transmission \cite{farsad2018deep,xie2021deep,zhou2021semantic,sana2022learning,lu2021reinforcement}, speech/audio transmission \cite{weng2021semantic,tong2021federated, han2022semantic}, and image/video transmission \cite{bourtsoulatze2019deep, kurka2020deep, zhang2022wireless,tung2022deep,jiang2022wireless}.

  However, it's noteworthy that most of the literature predominantly focuses on point-to-point communication scenarios, while shedding little light on exploring one-to-many semantic broadcast communication (BC). For example, \cite{ding2021snr} proposes to estimate the signal-to-noise ratio (SNR) at the decoder before adaptively decoding transmitted images. \cite{liu2022attention} extends the unidirectional and bidirectional decoders to a degraded broadcast channel, and \rxs~therein decode the text in different languages yet with the same meanings as needed. Similarly, \cite{hu2022one} designs a semantic recognizer DistilBERT to distinguish different receivers according to emotional properties (\ie, positive or negative). Different from \cite{ding2021snr,liu2022attention,hu2022one} that broadcast all extracted features, the semantic encoder in \cite{ma2023features} extracts disentangled semantic features and only broadcasts the receiver's intended semantic information to further improve transmission efficiency. Apparently, these efforts on semantic BC lay the very foundation for resource-efficient transmission to multiple receivers requiring the same content \cite{shrivastava20225g}.

  Besides, existing semantic BC schemes  \cite{ding2021snr,liu2022attention,hu2022one,ma2023features,shrivastava20225g} continuously adopt joint source-channel coding (JSCC) to optimize the end-to-end (E2E) deep neural networks (DNNs) encompassing multiple receivers. Hence, it is cost-ineffective to directly scale the JSCC framework to multi-user communication systems (e.g., semantic BC), due to the awful size of parameters corresponding to exponentially increased receivers. Instead, an alternative framework, which is built on top of a state-of-the-art transformer DNN structure (e.g., universal transformer \cite{dehghani2018universal}) and capably separates the training of \tx ~and \rxs, is highly essential. In this regard, alternate learning between \tx ~and \rxs~sounds intuitive, but its convergence could be questionable. Furthermore, in existing works, the objective function commonly leverages a differentiable function, such as cross-entropy (CE) or mean square error (MSE), so as to facilitate direct gradient backpropagation. This kind of bit-level metric, despite being feasible in most cases, inevitably introduces a ``semantic gap'' \cite{lu2022semantics} due to inconsistencies between evaluation metrics and optimization objectives, \ie, the training stage is optimized by the differentiable CE or MSE function, while the test stage is evaluated by non-differentiable semantic metrics, such as BLEU, CIDEr \cite{lu2022semantics}. Meanwhile, the difficulty of dealing with the non-differentiability of semantic similarity metrics could be exaggerated by the large dimension of sequence decoding space, which demands a simple algorithm. In addition, most of these SemCom approaches impractically assume the channel is known and differentiable to enable direct backpropagation through the channel. It remains challenging to solve the non-differentiability issue in SemCom.

  In summary, the design of a semantic BC with one \tx ~and multiple $\rm{RXs}$ is still hindered by the following critical issues:   
\begin{enumerate}
     \item [\textit{Q1:}] \textit{(Scalability) How to design a cost-effective semantic BC framework? }
    \item [\textit{Q2:}] \textit{(Compatibility) How to compatibly deal with the non-differentiability of semantic similarity metrics as well as practical channel propagation, so as to develop a learning-efficient algorithm? } 
    \item [\textit{Q3:}] \textit{(Convergence) How to establish a theoretical guidance for alternate learning and calibrate the convergent hyperparameters correspondingly? }
\end{enumerate} 

    \begin{table*}[hb]
		\centering
		\caption{Summary and comparison of semantic BC schemes}    
		\label{comparison-BC}

		\begin{tabular}{m{1cm} m{7cm} m{7cm} } 		
			\hline
			Refs&Brief description&Limitations\\ 
			\hline
                \cite{ding2021snr}&It is an SNR-adaptive deep JSCC scheme to minimize image distortion for multiple users. &The objective function is bit level MSE function that lacks semantic level supervision and the JSCC-based optimization framework is not easy to scale to more users.  \\
			\cite{liu2022attention}&It extends the proposed unidirectional and bidirectional decoders to the degraded broadcast channel. & The broadcast case is only designed for two users and lacks scalability.\\ 
                \cite{hu2022one}&It leverages users' semantic features to receive the desired information. & The adaptability of different channels has not been discussed and analyzed.\\
			\cite{ma2023features}&It is a feature-disentangled semantic BC framework. &The theoretical convergence has not been further analyzed. \\  	
			\hline
		\end{tabular}
\end{table*}
  In this paper, we jointly address these critical yet challenging issues and propose a convergent semantic BC framework based on self-critical alternate learning (SCAL), named SemanticBC-SCAL. In particular, SemanticBC-SCAL capably adapts to different channel environments from one \tx~ to multiple \rxs ~and competently solves the non-differentiability issue in the objective function (i.e., semantic similarity metrics) by incorporating RL techniques that regard the semantic similarity as the reward. Meanwhile, considering the difficulty of handling large semantic decision space, SemanticBC-SCAL utilizes a lightweight self-critical supervision \cite{luo2020better}, which capably provides a simple and expedited baseline estimation with low variance,  
  to yield a stable and efficient gradient policy. Furthermore, SemanticBC-SCAL is equipped with an asynchronous alternate learning mechanism to unify \tx~\& \rxs, swiftly adapting to different numbers of \rxs, and also addresses the non-differentiable channel. The primary contributions of our work are summarized as follows: 
 
\begin{itemize}
    \item We propose a semantic BC framework, named SemanticBC-SCAL, comprising one \tx ~and multiple \rxs, built upon the universal transformer (UT) architecture \cite{dehghani2018universal}, the superiority of which has been widely validated in \cite{xie2023semantic,zhou2021semantic}. Notably, the decoders in SemanticBC-SCAL share a consistent structure with trained parameters, which could be leveraged for further adaptive learning. Thus, it allows the trained decoder model to readily expand to additional \rxs ~at a reasonable cost, effectively addressing the \textit{Q1}.  

    \item Within this proposed SemanticBC-SCAL framework, in order to enable stable optimization via a non-differentiable semantic similarity metric, we regard semantic similarity as a reward and formulate it as an RL problem, which is compliant with any semantic metric. Meanwhile, we adopt self-critical supervision to provide simple and efficient gradient estimation. This becomes especially suitable for complex SemCom systems with large semantic decision space and partially addresses the mentioned \textit{Q2}.  

    \item Combining the above approaches, a self-critical based alternate learning mechanism is devised to alternately train the encoder at \tx~ and multiple decoders at \rxs, as the conventional supervised training can not directly perform backpropagation through non-differentiable channels.
     In particular, under self-critical supervision, when the \tx ~remains fixed, multiple \rxs ~locally update their parameters according to their unique channel properties. After certain iterations, the \rxs ~are fixed, and the \tx ~updates itself by aggregating the information of multiple \rxs ~uniformly. This alternate learning mechanism just addresses the non-differentiable channel in \textit{Q2}. Moreover, we carefully investigate the theoretical convergence and validate the performance. The results show that the \rxs ~can adapt to different channels, and yield superior performance compared to traditional communication method \cite{reed1960polynomial} and JSCC-based semantic model \cite{xie2021deep}. This addresses the mentioned \textit{Q3}. 
 
\end{itemize}
 
 The rest of this paper is mainly organized as follows. Section \ref{sec: Related Works} briefly explains the recent works and clarifies the novelty of our work. In Section \ref{sec: System Model and Problem Formulation}, we introduce the preliminaries of the semantic BC model and formulate the optimization problem. In Section \ref{sec: SemanticBC-SCAL Scheme with Alternating Update Mechanism}, we describe the alternate learning mechanism and self-critical supervision of the SemanticBC-SCAL framework and analyze the theoretical convergence. In Section \ref{sec: Performance Evaluation}, the numerical results are illustrated and discussed. Finally, 
 conclusions are drawn in Section \ref{sec: Conclusions}.

\section{Related Works }
\label{sec: Related Works}
  SemCom primarily concentrates on the source information reconstruction based on the received information. For text-based SemCom systems, Farsad \etal \cite{farsad2018deep} propose the initial JSCC based on the recurrent neural network (RNN) codecs for fixed length sentences, which shows that the DL-based codecs are capable of achieving lower word error rate and preserve semantic information by embedding sentences into a semantic space. Inspired by this success, Xie \etal \cite{xie2021deep} design a transformer-based SemCom system, named DeepSC, which aims to recover the semantics for arbitrary lengths of sentences under varying channels. Since then, SemCom has attracted a lot of attention and many works utilize JSCC frameworks to transmit sentence semantics \cite{zhou2021semantic, xie2020lite,sana2022learning}, such as UT \cite{dehghani2018universal} based semantic coding system \cite{zhou2021semantic,xie2023semantic}, and mutual information (MI) based performance optimization between semantic compression and semantic fidelity. 
  Moreover, Wang \etal \cite{wang2022performance} and Seo \etal \cite{seo2021semantics} adopt knowledge graph (KG) and probability distribution respectively to infer semantics information explicitly.
  In addition to text, Weng \etal \cite{weng2021semantic} design an attention-based end-to-end (E2E) SemCom system for speech transmission. Tong \cite{tong2021federated} \etal develop a federated learning (FL) trained model to deliver audio semantics between multiple devices and the server. Besides, Bourtsoulatze \etal \cite{bourtsoulatze2019deep} initially extend the JSCC framework and apply it to the deep image SemCom system, which provides graceful degradation as SNR changes. Afterwards, Kurka \etal in \cite{kurka2020deep} consider the channel feedback to enhance the image reconstruction quality based on the JSCC scheme. 
  Moreover, the authors in \cite{tung2022deep} and \cite{jiang2022wireless} introduce the semantics to maximize overall visual quality by transmitting some key motion points, so as to further improve transmission efficiency.

  Apart from the achievements in semantic level SemCom, some results are achieved at the effectiveness level as well, so as to transmit essential semantics timely for successful tasks execution. For example, Jankowski \etal \cite{jankowski2020deep,jankowski2020wireless} target the identification of individuals or vehicles and design two alternative (digital and analog communications-based) schemes to enhance retrieval accuracy. 
  Kountouris \etal \cite{kountouris2021semantics} develop a communication paradigm, in which 
  smart devices sample the ``semantics of information''  to steer their traffic for the purpose of real-time remote actuation.
  Additionally, in\cite{xie2022task}, a transformer-based framework is designed to unify the structure of the transmitter for different tasks with different types of data (single-modal data and multi-modal data), wherein many-to-many and many-to-one scenarios are exemplified. Similarly, \cite{zhang2022unified} also presents a unified E2E framework that can serve different tasks with multiple modalities by dynamically adjusting feature dimensions and DNN layers.  
  Commonly, existing works in SemCom focus more on point-to-point or multiple-encoder-single-decoder JSCC frameworks. Furthermore, as explained earlier, there emerge several works \cite{ding2021snr,liu2022attention,hu2022one,ma2023features} on Semantic BC, and the related drawbacks are summarized in Table \ref{comparison-BC}. Notably, these semantic BC works tend to encounter challenges related to scalability, compatibility, and convergence issues.

\begin{table}
    \centering
    \caption{Notions used in this paper.}
    \label{notation}
    \begin{tabular}{m{2.3cm} m{5.5cm}}
    \toprule   
    Notation      & Definition       \\
    \midrule
    $N$ & The number of receivers\\
    $\FTX, \FRXN$ &  The abbreviated terms for transmitter and receiver $n$ ($n\sim N$) respectively\\ 
    ${\boldsymbol{m}}, {{{\boldsymbol{\hat m}}}_n}$  &   Input message of \tx~ and decoded message of \rxn  \\
    $\Theta $  &  Semantic similarity metric \\ 
    ${\theta _{{\rm{en}}}}, {\theta _{\rm{q}}}$ & Parameters for semantic encoder and quantization layer\\
    ${\phi _{n,{\rm{de}}}}, {\phi _{n,{\rm{dq}}}}$ & Parameters for semantic decoder and dequantization layer\\
    $S{C_{{\rm{en}}}}( \cdot ), S{C_{n{\rm{,de}}}}( \cdot )$ & Semantic encoder and semantic decoder $n$ respectively \\
    $S{C_{{\rm{TX}}}}( \cdot ), S{C_{n{\rm{,RX}}}}( \cdot )$ & Semantic coding for \tx, ~and semantic decoding for \rxn~  \\
    $Q( \cdot ), Q_n^{ - 1}( \cdot )$ & Quantization and dequantization layer respectively \\
    $\theta, {\phi _n}$&  Parameters for \tx~ and \rxn ~respectively \\
    $T$, ${{\hat T}_n}$ & The length of input message for \tx~ and output message for \rxn ~respectively\\
    $D$ & Dimension of message embedding \\
    ${{\cal H}_n}( \cdot )$  & Random channel\\
    ${\boldsymbol{b}}, {{{\boldsymbol{\hat b}}}_n}$ &Quantization coding and dequantization coding for \rxn\\
    $B$ & Quantified bits for each word\\
    ${W_m}$  & The dictionary\\ 
    $V$  & Dimension of the dictionary\\ 
    ${p_n}( \cdot )$ & Probability of choosing a single word of semantic decoder $n$\\
    $K$ & Number of parallel samples\\
    ${{{\boldsymbol{\hat m}}}_{n,i}}$ & The $i$-$\rm{th}$ complete decoding trajectory for \rxn \\
    ${r_n^{(t)}}$ & Reward at step $t$ for semantic decoder $n$\\
    ${s_n^{(t)}}$ & State at step $t$ for semantic decoder $n$\\
    $s_{n,{\rm{de}}}^{(t)}$ & State at step $t$ for semantic decoder $n$\\
    $\mathcal{S}, \mathcal{A}$ & The state space and action space\\
    $P$ & State transition probability function\\
    $R$ & Reward space \\ 
    $\gamma $ & Discounted factor \\ 
    $\kappa$ & Local iterations for decoders in each cycle \\
    \bottomrule
    \end{tabular}
\end{table}

 In addition, there exist some efforts to combine RL and SemCom to improve the communication efficiency of multi-agent systems \cite{yuan2021graphcomm, yun2021attention,tung2021effective,wang2022performance,li2021task}. 
 For example, in order to alleviate information bottleneck (IB) and improve the effectiveness of the multi-round communication process, \cite{yuan2021graphcomm} applies an attention-based sampling mechanism to the graph convolution network for multi-agent cooperation.  
 Different from the aforementioned works that utilize SemCom for effective transmission, \cite{li2021task} considers an RL based semantic bit allocation to implement task-driven semantic coding for a traditional hybrid coding framework. Specifically, they design semantic maps for different tasks to extract the pixel-wise semantic fidelity and leverage RL to integrate the semantic fidelity metric into the in-loop optimization of semantic coding. It can be seen that RL is capable of improving communication efficiency and has demonstrated considerable advantages. Nevertheless, there are relatively few works using RL to develop an in-depth optimization for the SemCom system, despite its potential effectiveness in addressing scalability and compatibility issues.

 \section{System Model and Problem Formulation} 
 \label{sec: System Model and Problem Formulation}  
In this section, we introduce the basic UT \cite{dehghani2018universal} based semantic BC model and also give a brief description of related semantic metrics for performance evaluation. Beforehand, we summarize mainly used notations in Table \ref{notation}. 

\subsection{System Model} 
\label{sec: system_model}

As illustrated in Fig. \ref{architecture}, the proposed BC system mainly includes three parts, that is, one \tx, $N$ random channels and $N$ \rxs. Generally speaking, at \tx ~side, the message ${\boldsymbol{m}}$ is first converted to an embedding by the embedding layer. Then the embedding is sent to implement semantic encoding, including a semantic encoder for features extraction and a dense layer for channel encoding. Then, the quantization layer quantifies ${\boldsymbol{x}} $ as $\boldsymbol{b}$ to facilitate the subsequent transmission and this extracted valuable information is broadcasted through different physical channels to reach \rxn ~($n\sim N$)\footnote{The $\sim$ operator implies $n \in \{ 1,2,\cdots N \}$.}.  
Likewise, at \rxn ~side, the dequantization layer detects and reconstructs the received symbol ${{\boldsymbol{\hat b}}_n}$ from the received noisy signals, 
and the semantic decoding, \ie, including channel decoder and semantic decoder, estimate and recover the source as ${{{\boldsymbol{\hat m}}}_n}$ based on the knowledge base (KB). 
Different from the classic SemCom \cite{xie2021deep}, we adopt the UT \cite{dehghani2018universal} based semantic encoder/decoders to achieve semantic coding/decoding. Besides, \rxs ~share the same structure but have individual DNN parameters accommodating to channel properties, so as to recover the original information as accurately as possible.

  	\begin{figure*}[ht]
		\centering
  \setlength{\abovecaptionskip}{-0.2cm}
		\includegraphics[width=0.935\linewidth]{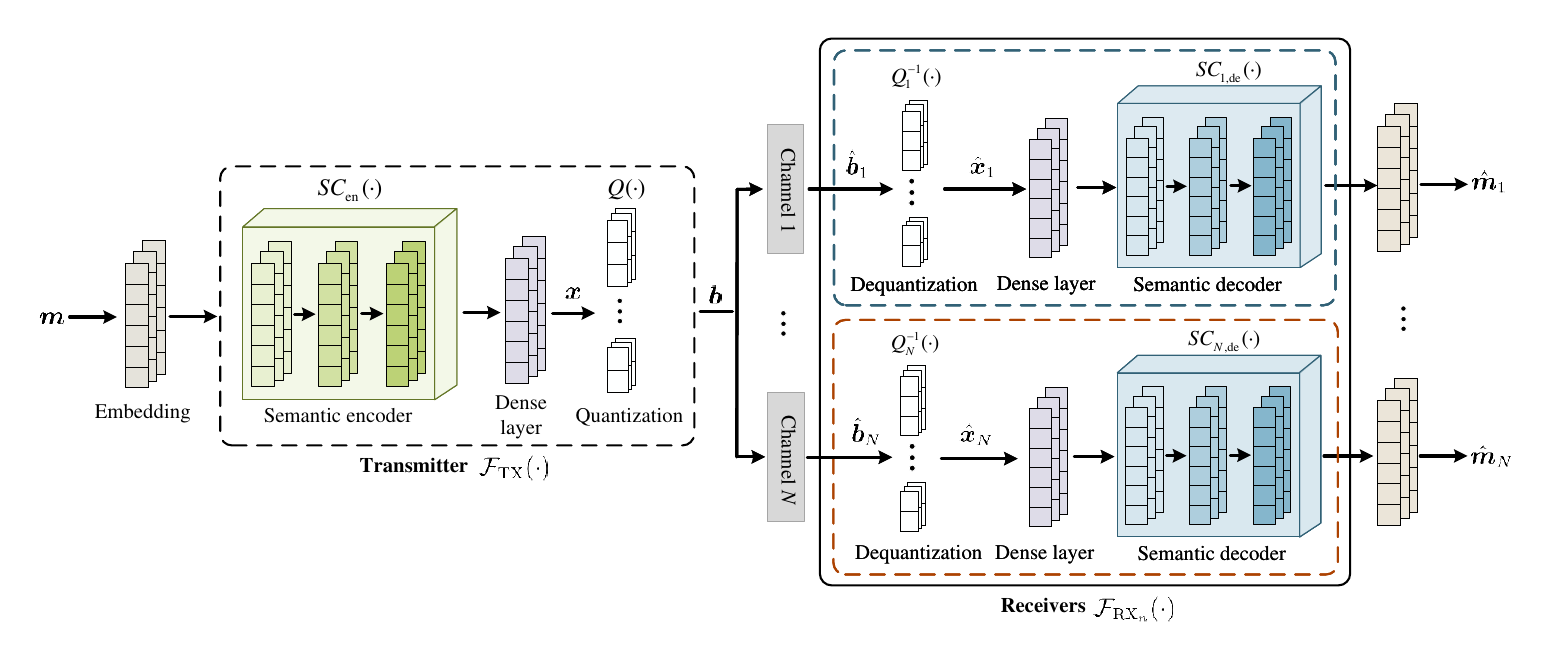}\\
		\caption{The system model for semantic BC system.} 
		\label{architecture} 
	\end{figure*}	
 
Mathematically, for a source message ${\boldsymbol{m}} = \{ {w^{(1)}},\cdots,{w^{(i)}},\cdots,{w^{(T)}}\} $, where $T$ is the sequence length, and ${w^{(i)}}$ can be regarded as a word from dictionary ${W_m}$ (\ie, $w^{(i)} \in W_m$). The encoding process $\FTX\left(  \cdot  \right)$ (parameterized by $\theta$) for \tx ~consists of semantic encoding $S{C_{{\rm{en}}}}\left(  \cdot  \right)$ and quantization $Q\left(  \cdot  \right)$. The semantic encoding is given by  
\begin{equation}
    \label{eq:encoding}
    {\boldsymbol{x}} = \emph{SC}_{{\rm{en}}}(\boldsymbol{m} ;{\theta_{\rm{en}}})
\end{equation}
where ${\boldsymbol{x}} \in {\mathbb{R}}^{T \times D}$, $D$ is the dimension of each symbol in ${\boldsymbol{x}}$, $S{C_{{\rm{en}}}}( \cdot )$ is the combination of semantic encoder and channel encoder parameterized by $\theta_{\rm{en}}$.

After that, the embedding ${\boldsymbol{x}} $ is quantified into discrete bits by a $\theta _{\rm{q}}$-induced quantization $Q( \cdot )$ as

\begin{equation}
    \label{eq:quantization} 
    {\boldsymbol{b}} = Q({\boldsymbol{x}};{\theta _{\rm{q}}})  
\end{equation}
where ${\boldsymbol b} \in {\mathbb{R}}^{T \times B}$, $B$ is the number of bits for each quantified symbol of ${\boldsymbol{x}} $. Then \tx ~broadcasts ${\boldsymbol{b}}$ to multiple physical channels. 

At the \rxn ~side, different from the point-to-point transmission with a single receiver, the transmitted bits in semantic BC experience heterogeneous channel conditions corresponding to different receivers.  
Analogously, the decoding process $\FRXN (\cdot)$ (parameterized by $\phi_n$) for \rxn ~involves the reverse steps, including dequantization and semantic decoding, which are represented as $Q_n^{ - 1}( \cdot )$ and $S{C_{n,{\rm{de}}}}( \cdot )$ respectively. With the aid of channel state information (CSI), the semantic decoder can be dynamically adjusted to achieve adaptive decoding. As such, the decoding message ${{\boldsymbol{\hat m}}_n}$ of \rxn ~can be denoted as:  
\begin{equation}
    \label{eq:decoding}
    {{{\boldsymbol{\hat m}}}_n} = S{C_{n,{\rm{de}}}}\left( {{Q_n^{ - 1}}({{\cal H}_n}({\boldsymbol{b}});{\phi _{n,{\rm{dq}}}});{\phi _{n,{\rm{de}}}}} \right)
\end{equation}
 where ${{\cal H}_n}( \cdot )$ ($n\sim N$) is the channel layer and is usually modeled as additive white Gaussian noise (AWGN) channel or multiplicative Rayleigh fading channel \cite{bourtsoulatze2019deep}.

   As such, \rxn ~recovers the source message by predicting the probability distribution of the next word in the dictionary. In the end, a complete decoded message can be denoted as $ {{{\boldsymbol{\hat m}}}_n}=\{\hat w_n^{(1)},\hat w_n^{(2)},\cdots,\hat w_n^{{({\hat T_n}})}\}$, where ${{\hat T_n}}$ is the length of decoded message for \rxn.

 \subsection{Problem Formulation}  
 \label{sec: Problem Formulation}

    Given the described system model, our goal is to maximize semantic similarity for \rxn ~by learning the optimal parameters for the system $\left\langle {{{\cal F}_{{\rm{TX}}}},{{\cal F}_{{\rm{R}}{{\rm{X}}_n}}}} \right\rangle $, that is,
    
\begin{equation}
    \label{eq: optimal system}
    \left\langle {{\theta ^ * },\phi _n^ * } \right\rangle  = \mathop {\arg \max }\limits_{{{\cal F}_{{\rm{TX}}}},{{\cal F}_{{\rm{R}}{{\rm{X}}_n}}}} \Theta ({\boldsymbol{m}},{{{\boldsymbol{\hat m}}}_n})
\end{equation}
where $\Theta ( \cdot )$ can be any reasonable semantic similarity metric, and we do not assume its differentiability. Therefore, in this paper,  
     bilingual evaluation understudy (BLEU) scores \cite{papineni2002bleu}, BERT (bidirectional encoder representation from transformers) \cite{devlin2018bert} based similarity metric (i.e., BERT-SIM) and word accuracy rate (WAR) are adopted to measure the degree of consistency between the input $\boldsymbol{m}$ and output ${{{\hat{\boldsymbol m}}}_n}$ \cite{xie2021deep, lu2022semantics}.   

    Nevertheless, the problem defined in \eqref{eq: optimal system} for semantic BC is challenging to solve by traditional JSCC frameworks. The difficulties are twofold. On one hand, in order to address the inherent ``semantic gap'', the commonly used differentiable objective will no longer be applicable since the gradient cannot be directly propagated backwards through the channel. 
    On the other hand,  along with the number of \rxs ~ increases, the scalability issue emerges, primarily as the commonly used point-to-point optimization method becomes invalid and untenable. Furthermore, to alleviate the ``semantic gap'', we formulate \eqref{eq: optimal system} as an RL problem by regarding non-differentiable semantic similarity as a reward. Besides, we adopt an alternate learning mechanism to cope with the scalability issue by iteratively updating the encoder and decoders, in which the encoder and decoders are all conceptualized as independent agents. 

    \textit{Alternate learning mechanism:} Specially, upon freezing the parameters of $\FTX$ (\ie, $\theta $), $\FRXN$ is locally updated $\kappa $ iterations. Conversely, when $\FRXN$ (\ie, ${\phi _n}$) is held constant, $\FTX$ is updated once independently by averaging the decoders' results. It is worth emphasizing that the encoder updates once, while each
    decoder locally updates $\kappa$ iterations, collectively forming a single \textit{update cycle}.
    Corresponding to this alternate learning mechanism, the objective functions for the \tx ~and \rxn ~diverge intrinsically. In this regard, the overall optimization objective outlined in \eqref{eq: optimal system} is formulated into the following two parts    
\begin{subequations} 
\label{eq: optimal}
	\begin{align}
            &{\theta ^ * } = \mathop {\arg \max }\limits_{\FTX( \cdot )} \Theta ({\boldsymbol{m}},{\rm{avg}}(\underbrace {\FRXN(}_{{\rm{no~grad}}} {\cal H}_n(\FTX({\boldsymbol{m}})))))\label{eq:encoder_objective}\\ 
		&\phi _n^ *  = \mathop {\arg \max }\limits_{\FRXN( \cdot )} \Theta ({\boldsymbol{m}},\FRXN( {\cal H}_n (\underbrace {\FTX({\boldsymbol{m}})}_{{\rm{no~grad}}})))  \label{eq:decoder_objective}            
	\end{align}
\end{subequations}
where $\rm{avg}(\cdot)$ is to get the mean value.

In this paper, we talk about how to alternately optimize (\ref{eq: optimal}a) and (\ref{eq: optimal}b) on top of RL, the details of which shall be presented in the next section.

 \section{SemanticBC-SCAL Scheme with Alternating Learning Mechanism } 
 \label{sec: SemanticBC-SCAL Scheme with Alternating Update Mechanism}
 In this section, we first elaborate on how to formulate a Markov decision process (MDP) to facilitate the application of RL, optimizing the semantic BC system with semantic level supervision. Afterwards, we talk about the self-critical alternate learning based algorithm, SemanticBC-SCAL, to produce a solution. Furthermore, we also investigate its convergence.  
 \subsection{The Markov Decision Process (MDP) Framework}  
 \label{sec: MDP}
 
 Following this multi-user semantic BC system, where the encoder can be formulated as one agent at \tx ~and the decoders can be modelled as $N$ independent agents at \rxs, which interacts with the environment (\ie ~the text dictionary $W_m$) by taking actions on the basis of states and receiving the rewards from the environment.  
 
 Without loss of generality, the MDP setting can be denoted as $(\mathcal{S},\mathcal{A}, P, R,\gamma)$, where $\mathcal{S}$ and $\mathcal{A}$ is the state space and the action space respectively, $P$ is the state transition probability, and $R$ denotes reward function with discounted factor $\gamma  \in (0,1]$. 
 Since the actions (\ie, word generation process) are implemented at the decoders, the state space and action space are determined by \rxn~and shared for the encoder. As such, we define the state space as $\mathcal{S} = \{ {s_1},\cdots, {s_n},\cdots, {s_N}\}, n\sim N$ and action space (\ie, vocabulary) as $\mathcal{A}=\{{w_1},\cdots, {w_n},\cdots, {w_N}\},{w_n} \in  W_m$ (where the dimension of dictionary ${W_m}$ is $V$). Moreover, the state transition function as $P = \{{P_{{\pi _{{\theta}}}}}, {P_{{\pi _{{\phi _n}}}}} \}$, and the reward for \tx ~and \rxn ~ are distinctly parameterized by ${\pi _\theta }$ and ${\pi _{{\phi _n}}}$ respectively, which are further delineated as $ R = \{ {R_{{\pi _\theta }}},{R_{{\pi _{{\phi _n}}}}}\} $.

    In particular, the MDP can be detailed as follows. 
\begin{itemize}
     \item \emph{Action}. The definition of action is to generate the next token $a_n^{(t)} = \hat w_n^{(t)}$ from the dictionary $W_m$ at \rxn, where the sequence generation process starts with token ``<SOS>'' and ends with ``<EOS>''. In contrast to traditional communication where there exists a one-to-one relationship among the source message $\boldsymbol{m}$, SemCom introduces a scenario where multiple pairs of interpretations can exist between them. Hence, the state-action space can be larger than traditional methods. 
    \item \emph{State.} We define the state as the sequential state of decoder $n$ at step $t$ and decoded tokens (\ie, words in the text tasks) so far. Furthermore, for each \rxn, $s_{n,{\rm{de}}}^{(t)}$ can be viewed as a combination of previous decoder output and buffer generation (words) so far, \ie, $s_n^{(t)} = \{ s_{n,{\rm{de}}}^{(t)},\hat w_n^{(1)},\hat w_n^{(2)},\cdots,\hat w_n^{(t)}\} \in \mathcal{S} $. It merits emphasis that our goal is to interpret messages at the semantic level, which may lead to variations in length between the decoded sentences and their original counterparts. Here, for different \rxn, the maximum length of decoded sentence is ${{\hat T}_n}$.  
    Meanwhile, once the next token $\hat w_n^{(t + 1)}$ is generated, the state transition is deterministic between any adjacent state.    
    \item \emph{Policy.} Notably, the policies for \tx ~and \rxs~ shall be independently calibrated, since the semantic encoder at \tx ~aims to optimize the encoding process in \eqref{eq:encoder_objective}, while the semantic decoder at \rxn ~is responsible for generating semantically accurate messages by maximizing \eqref{eq:decoder_objective}. 
   
    Specifically, for \tx ~side, the semantic encoding of the encoder is a continuous one and only needs to be implemented once for each sample. 
    In this sense, we model this process as a continuous Gaussian distribution with the mean value $\boldsymbol{\mu} ={\FTX(\boldsymbol{m})} \in \mathbb{R}^{1 \times D}$, and the covariance matrix $\Sigma ={(\sigma \boldsymbol{I})^2}\in \mathbb{R}^{D\times D}$ (where $\sigma $ is typically set to 0.1). In other words, some Gaussian noise is intentionally added to the semantic encoding process, to encourage a certain  exploration. Mathematically, the encoder policy can be written as  
\begin{equation}
    \label{eq:encoder-strategy}  
		\begin{aligned}                         
              {\pi _{\theta }}: = \text{Sample} \left({\cal N} \left(\boldsymbol{\mu } = {\FTX(\boldsymbol{m})},\Sigma  = {(\sigma {\boldsymbol{I}})^2} \right)\right)
		\end{aligned}
\end{equation}
    where $\mathcal{N}$ indicates a Gaussian distribution.

    For \rxn, instead of one deterministic distribution, the input of the decoder is decoded by $K$ times and the output likelihood is modelled as a probabilistic multinomial distribution (termed as $ \mathcal {MD}$) under self-critical supervision (it will be detailed in Section \ref{sec: Self-critical Optimization under Alternate Learning Mechanism}), which means \rxn ~can obtain a group of different output from $K$ views, so as to better reap the potentially diversified semantics of one same sentence. As such, sampling from this kind of multinomial distribution ensures the collection of samples with the maximum likelihood from multi-view. 
    For each multinomial distribution, it is modelled by the ``softmax'' function that is parameterized on the ${\phi _n}$. Then for step $t$, one has 
\begin{equation}
    \label{eq:decoder-strategy}
               \pi _{{\phi _n}}^t: = \text{Sample} ( \mathcal {MD}([p_n(\hat w_1^{(t)}),p_n(\hat w_2^{(t)}),\cdots,p_n(\hat w_V^{(t)})]))
\end{equation}
where the superscript $t$ means the token time-step for sampling operation.

    \item \emph{State transition probability.} According to the above-mentioned definitions, the state transition function $P$ is deterministic between two adjacent states. Formally, we have  ${P}({s_{n}^{(t + 1)}|s_{n}^{(t)},\hat w_n^{(t)}} ) = 1$.  

    \item \emph{Reward.} We use the semantic similarity as a reward and denoted as $\Theta ({\boldsymbol{m}},{{\boldsymbol{\hat m}}_n})$. In the sentence generation model, the reward can only be obtained until the last token is generated. Therefore, the reward is sparse in SemanticBC-SCAL. In other words, the intermediate reward is always zero except for the last time step ${{\hat T}_n}$, that is,
\begin{equation}
    \label{eq:reward}
    	{r_n^{(t)}} = \left\{
		\begin{aligned}
			0\quad\quad		                    &\quad \textrm{if}  \   t \neq {{\hat T}_n} \\
			\Theta ({\boldsymbol{m}},{{\boldsymbol{\hat m}}_n}) 	 	&\quad \textrm{if} \   t = {{\hat T}_n} 
		\end{aligned}
		\right. 		
\end{equation}   
    
\end{itemize}

For the sake of simplification, we assume the discounted factor $\gamma  = 1$, so without loss of generality, the \textit{return} for our system is formulated as: 
\begin{equation}
    \label{eq:accumulative return}
        G_n^{(t)} = \sum\limits_{k = t + 1}^{{{\hat T}_n}} {{\gamma ^{k - t - 1}}r_n^{(k)}}= \sum\limits_{t = 1}^{{{\hat T}_n}} {r_n^{(t)}}  		 		
\end{equation}

\begin{figure*}[hbt]
   \centering
   \setlength{\abovecaptionskip}{-0.2cm}
		\includegraphics[width=0.88\linewidth]{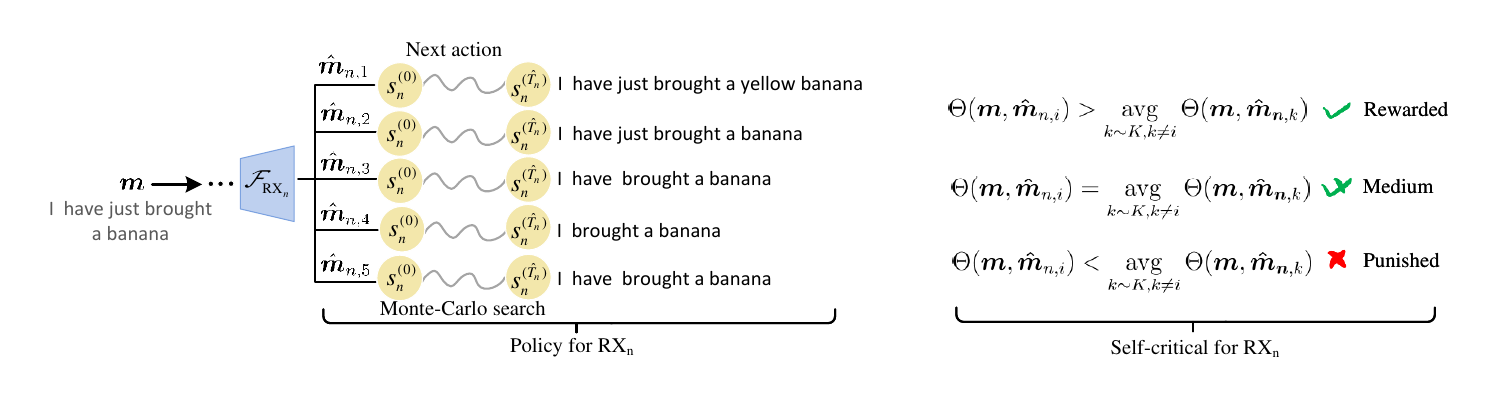}\\ 
		\caption{Illustration of the self-critic optimization for decoder side at \rxn.}
		\label{self-critic}
\end{figure*}

Accordingly, the state-value function is given by $V(s_n^{(t)}) = \mathbb{E}\left[ {G_n^{(t)}|s_n^{(t)}} \right]$. 
Consistent with the reformulated problem in (\ref{eq: optimal}), at decoder sides, the objective function for \rxn ~is to maximize the value function, which can be viewed as maximizing \textit{return} for a complete trajectory $\{{\hat w_n^{(t)},s_n^{(t + 1)}, \hat w_n^{(t+1)},\cdots}\}$ under decoder policy ${\pi _{{\phi _n}}}$ starting from $s_n^{(0)}$, which is given by 
\begin{equation}
    \label{eq: objective function decoder}		                                     
   J({\phi _n}) = {V_{{\pi _{{\phi _n}}}}}(s_n^{(0)}) 
  =  {\mathbb{E}_{_{{\pi _{{\phi _n}}}}}}\left[ {G_n^0|s_n^0} \right]
   = {\mathbb{E}_{_{{\pi _{{\phi _n}}}}}}\left[ {\sum\limits_{t = 1}^{{{\hat T}_n}} {r_n^{(t)}} } \right] 	
\end{equation}

\noindent While for encoder optimization, \tx~ aims to find an optimal encoding to maximize the semantic accuracy for all decoders. Under the encoder policy ${\pi _{{\theta}}}$, we regard the maximum average \textit{return} of all decoders as the objective function, that is 	                        \begin{equation}
               J({\theta}) = {1 \over N}\sum\limits_{n = 1}^N {{V_{{\pi _{\theta  }} }} (s_n^{(0)})} 
               ={1 \over N}\sum\limits_{n = 1}^N { \mathop {\mathbb{E}_{{\pi _{{\theta}}}}} \left[ {\sum\limits_{t = 1}^{{\hat T_n}} {r_n^{(t)}}} \right] } 
 \label{eq: objective function encoder}
\end{equation}

\begin{figure*}[hbt]
   \centering
   \setlength{\abovecaptionskip}{-0.2cm}
		\includegraphics[width=0.88\linewidth]{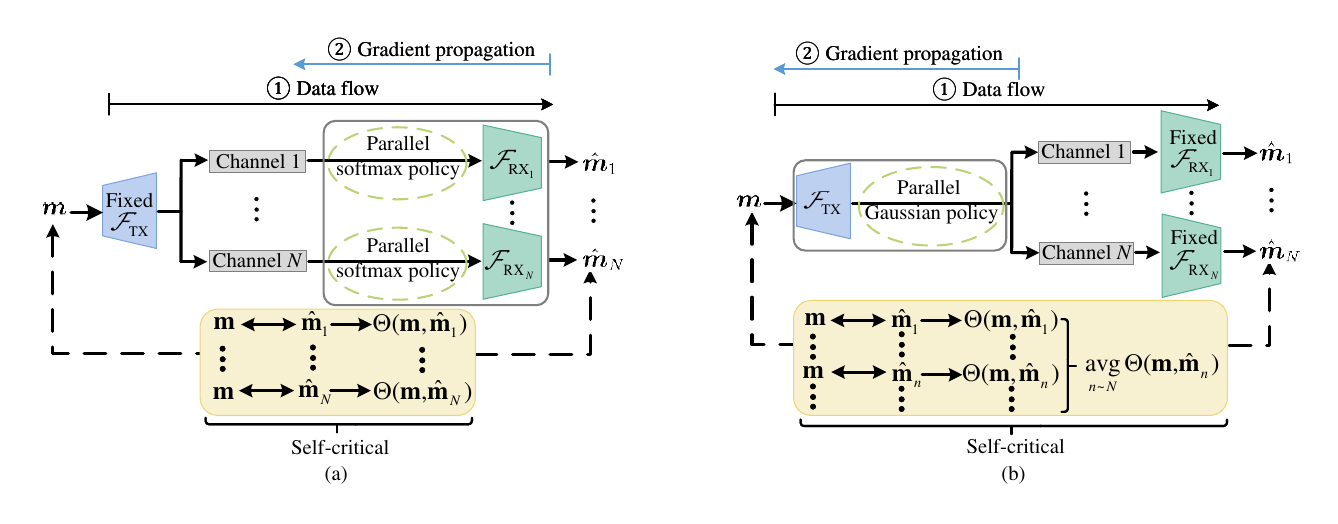}\\
		\caption{Optimization process under the alternate learning mechanism. (a) The training process of \rxn. (b) The training process of \tx. (The black solid line is the direction of data flow, while the blue solid line denotes the direction of gradient propagation.)}
		\label{Training process} 
\end{figure*}   
 \subsection{Self-Critical Optimization under Alternate Learning Mechanism}   
\label{sec: Self-critical Optimization under Alternate Learning Mechanism}

 In order to find a proper solution to maximize the objectives functions (\ref{eq: objective function decoder}) and (\ref{eq: objective function encoder}), one most straightforward solution is to simulate the environment to obtain Monte-Carlo trajectory $\{ (s_n^{(t)},\hat w_n^{(t)},r_n^{(t)})\} _{t = 1}^{\hat T_n}$ for each time step $t$ \cite{gao2019self}. Though Monte-Carlo rollout provides an unbiased estimation for expected \textit{return}, large state-action space unavoidably leads to high variance \cite{dubey2016variance}. Another alternative solution is the ``actor-critic'' method and its variants\cite{dogru2021actor}, in which the baseline is established by constructing another \emph{value network} to criticize the policy \cite{lu2021reinforcement}. However, this way will introduce extra parameters and estimation bias. Meanwhile, due to the extensive number of tokens included in the dictionary $W_m$, the size of state-action space exceeds ${10^{4}}$. In this sense, merely the Monte-Carlo trajectory or actor-critic methods are impractical to implement for sentence generation in the SemCom system.  

 Based on the aforementioned discussions, to handle this large-scale sentence generation, we develop the SemanticBC-SCAL algorithm to combine a simper and practical RL algorithm, \ie, self-critical learning \cite{luo2020better}, with the abovementioned alternate learning mechanism in Section \ref{sec: Problem Formulation}. Specially, SCAL samples a group of parallel Monte Carlo trajectories and utilizes the mean reward from the rest of $K-1$ parallel samples, that is,  ${\boldsymbol{\hat m}}_n=\{ {{{\boldsymbol{\hat m}}}_{n,1}},\cdots,{{{\boldsymbol{\hat m}}}_{n, K-1}}\}$ (where $K$ is the number of samples), to act as the baseline function. To generate a complete trajectory by introducing a certain exploration, the semantic encoder/decoders could employ the policy  ${\pi _\theta }$ and ${\pi _{{\phi _n}}}$ respectively to guide the system in accommodating multi-view semantic representations. Since the baseline comes from itself, \ie, its other $K-1$ decoding results, this kind of supervision can be regarded as ``self-critical'' manner \cite{rennie2017self,gao2019self, luo2020better}, and the empirical results show that larger $K$ will have a larger absolute value of the mean and introduce lower variance \cite{gao2019self}. 
 Taking an example of decoder optimization in Fig. \ref{self-critic}, the $\rm{TX}$ sends the message (``I  have just brought a banana'') to the broadcast channel, and the $\rm{RX}_n$ first samples $K=5$ parallel samples before channel decoding, in which the decoded parallel sentences could be different. Then, \rxn ~can be viewed as an actor responsible for taking actions by Monte Carlo sampling to generate $5$ parallel decoding results. As such, the advantage of one decoding result will be evaluated by the other 4 parallel samples. 
 
  Furthermore, combining the self-critical learning and alternate learning mechanism, a learning-efficient method is devised to asynchronously update the one \tx ~and multiple \rxs ~to adapt to its own channel environment. 
  To be specific, in the decoder training stage, we freeze the encoder parameters $\theta $ and locally update the decoder parameters $\phi_n$. Since the semantic similarity score can only be obtained via the self-critical manner after decoding a complete sentence, the decoder updating occurs only after a complete sentence transmission. The details of gradients and flow in SemanticBC-SCAL are illustrated in Fig. \ref{Training process}(a).  

  Moreover, we provide how to compute the gradient of (\ref{eq: objective function decoder}) by the following theorem.

\begin{theorem}\label{thm:theorem1}
(Self-critical semantic policy gradient for $\rm{RX}_n$) For any semantic decoder $n$ optimized by the alternate learning mechanism with $K$ parallel sampling, suppose that the decoder is updated by the policy  $\pi _{{\phi _n}}^{(t)}$ to generate a complete sequence $\{\hat w_n^{(1)},\hat w_n^{(2)},\cdots,\hat w_n^{({{\hat T_n}})}\}$, the gradient can be approximated by

\begin{equation}
\begin{aligned}
\label{eq: decoder-gradient0}
      {\nabla _{{\phi _n}}}J({\phi _n}) &\approx {1 \over K}\sum\nolimits_{i = 1}^K \left[ {\sum\nolimits_{t = 1}^{{{\hat T}_n}} {{\nabla _{{\phi _n}}}\log {\pi _{{\phi _{n,i}}}}(\hat w_n^{(t)}|s_n^{(t)})} } \right.\\      
		& \quad \quad \left.{ \cdot \left( {{\Theta _{n,i}} - \mathop {{\rm{avg}}}\limits_{k\sim K,k \ne i} {\Theta _{n,k}}} \right)} \right]\\
	\end{aligned} 
\end{equation}
where the subscript $i$ in ${\pi _{\phi _n, i} }$ denotes the i-th parallel sampling policy of sample $i$, while ${\rm{avg}}( \cdot )$ gets the mean value. ${\Theta _i}$ is the return of current sample $i$, and ${\Theta _k}$ represents any one of parallel samples. 
\end{theorem}
\begin{proof}

 During the decoder optimization process, a complete trajectory $ {{{\boldsymbol{\hat m}}}_n}=\{\hat w_n^{(1)},\hat w_n^{(2)},\cdots,\hat w_n^{({{\hat T_n}})}\}$ has the probability of $P({{{\boldsymbol{\hat m}}}_n}|{\phi _n})$, and the gradient of $\rm{RX}_n$ is 
\begin{align}
    \label{eq: decoder-gradient1}
		 & J({\phi _n}) \nonumber\\
   = &{\nabla _{{\phi _n}}}\left[ {\sum\nolimits_{{{{\boldsymbol{\hat m}}}_n}} {P({{{\boldsymbol{\hat m}}}_n}|{\phi _n})\sum\nolimits_{t = 1}^{{{\hat T}_n}} {r_n^{(t)}} } } \right]\nonumber\\
               =& \sum\nolimits_{{{{\boldsymbol{\hat m}}}_n}} {\left[ {{\nabla _{{\phi _n}}}P({{{\boldsymbol{\hat m}}}_n}|{\phi _n}) \cdot \sum\nolimits_{t = 1}^{{{\hat T}_n}} {r_n^{(t)}} } \right]}  \\
                \stackrel{(a)}{=} & \sum\nolimits_{{{{\boldsymbol{\hat m}}}_n}} {\left[ {P({{{\boldsymbol{\hat m}}}_n}|{\phi _n}){\nabla _{{\phi _n}}}\log \left( {P({{{\boldsymbol{\hat m}}}_n}|{\phi _n})} \right) \cdot \sum\nolimits_{t = 1}^{{{\hat T}_n}} {r_n^{(t)}} } \right]}\nonumber  \\
                 = & \mathbb{E}_{{{{\boldsymbol{\hat m}}}_n}} \left[ {{\nabla _{{\phi _n}}}\log \left( {P({{{\boldsymbol{\hat m}}}_n}|{\phi _n})} \right) \cdot \sum\nolimits_{t = 1}^{{{\hat T}_n}} {r_n^{(t)}} } \right] \nonumber
\end{align}
where the equality (a) is derived by using the log-trick $\nabla \log x = \nabla x/x$. Since we sample one Monte Carlo trajectory to estimate the ${{{{\boldsymbol{\hat m}}}_n}}$, then one obtains 
\begin{equation}
    \label{eq: decoder-gradient2}
        {\nabla _{{\phi _n}}}J({\phi _n}) \approx {\nabla _{{\phi _n}}}\log \left( {P({{{\boldsymbol{\hat m}}}_n}|{\phi _n})} \right) \cdot \sum\nolimits_{t = 1}^{{{\hat T}_n}} {r_n^{(t)}} 
\end{equation}

As the sentence is generated step by step through ${{{\hat T}_n}}+1$ actions, ${P({{{\boldsymbol{\hat m}}}_n}|{\phi _n})}$ is further written as $P({{\boldsymbol{\hat m}}_n}|{\phi _n}) = {p(s_n^0)\prod\nolimits_{t = 1}^{{{\hat T}_n}} {\left[ {{\pi _{{\phi _n}}}(\hat w_n^{(t)}|s_n^{(t)}){P_{{\pi _{{\phi _n}}}}(s_n^{(t)}|s_n^{(t-1 )},\hat w_n^{(t-1)}})} \right]} } $. Meanwhile, the reward function is sparse, \ie, $\sum\nolimits_{t = 1}^{{{\hat T}_n}} {r_n^{(t)}} = {\Theta ({\boldsymbol{m}},{{{\boldsymbol{\hat m}}}_n})}$, and the state transition probability ${P_{\pi _{{\phi _n}}}(s_n^{(t)}|s_n^{(t-1 )},\hat w_n^{(t-1)}})$ is deterministic between two adjacent states. Therefore, one has 
\begin{equation}
    \label{eq: decoder-gradient3}
       {\nabla _{{\phi _n}}}J(\theta ,{\phi _n})  \approx \sum\nolimits_{t = 1}^{{{\hat T}_n}} {{\nabla _{{\phi _n}}}\log {\pi _{{\phi _n}}}(\hat w_n^{(t)}|s_n^{(t)})}  \Theta ({\boldsymbol{m}},{{{\boldsymbol{\hat m}}}_n}) 
\end{equation}
  Finally, the self-critical baseline function is implemented by utilizing the mean reward from the rest of $K - 1$ parallel samples, thus, the proof ends with \eqref{eq: decoder-gradient0}. 
\end{proof} 

\noindent \textbf{Remark:} It is worth mentioning that the term $\bigg( {{\Theta _{n,i}} - \mathop {\rm{avg}}\limits_{k\sim K,k \ne i} {\Theta _{n,k}}} \bigg)$ is regarded as baseline or advantage of policy ${\pi _{\phi_n, i} }$, which is a success of SCAL because the critical samples come from its own model and also obtain a  gradient variance reduction \cite{luo2020better}. Furthermore, in semantics representation, the one with higher semantic similarity gives a larger positive reward, while those with lower advantage will be punished.  

 As for the encoder, recalling that as in Fig. \ref{Training process}(b), the encoder output is converted to $K$ parallel samples by Gaussian policy, which are then sent to broadcast channels. Therefore, the objective of the encoder is to maximize the average reward of all decoders, the average decoded semantic similarity of different decoders is conversely used to optimize the $\FTX$. Analogously, we give the optimization direction of (\ref{eq: objective function encoder}) by the following theorem.

\begin{theorem}\label{thm:theorem2}
(Self-critical semantic policy gradient for $\rm{TX}$) For semantic encoder optimized by the alternate learning mechanism, suppose the encoder is updated by 
the parallel Gaussian strategy $\pi _\theta$, the gradient can be approximated by the Monte-Carlo sampling as
\begin{align}
\nabla_\theta J(\theta) 
&\approx\frac1N\cdot\frac1K\sum\nolimits_{n=1}^N\sum\nolimits_{i=1}^K\left[\left[\widetilde{\mathcal{F}_{\rm{TX}}}(\mathbf{m})-\mathcal{F}_{\rm{TX}}(\mathbf{m})\right]^T\right. \label{eq: encoder-gradient1}  \\
&\quad \left.\cdot\boldsymbol{\Sigma}^{-1}\left[\nabla_{\theta}\mathcal{F}_{\rm{TX}}(\mathbf{m})\right]\cdot\left(\Theta_{n,i}-\underset{k\thicksim K,k\neq i}{\operatorname* 
{avg}}\Theta_{n,k}\right)\right]  \nonumber 
\end{align}
where $\widetilde{\mathcal{F}_{\rm{TX}}}$ is the Gaussian sampling for $\mathbf{m}$, $T$ is the length of source message $\boldsymbol{m}$.   
\end{theorem}
\begin{proof} 
The proof is similar to that for Theorem \ref{thm:theorem1}. Thus, we only provide a sketch here. 

Basically, for a Gaussian distribution
\begin{align}
     &\pi_{\theta}(\boldsymbol{x};\boldsymbol{\mu}, \boldsymbol{\Sigma}) \\
     &= \frac{1}{(2\pi)^{D/2} \sqrt{\det{\boldsymbol{\Sigma}}}} \exp \left(-\frac{1}{2}(\boldsymbol{x}-\boldsymbol{\mu})^T \boldsymbol{\Sigma}^{-1} (\boldsymbol{x}-\boldsymbol{\mu})\right) \nonumber
\end{align}
we have
\begin{equation}
\label{eq: Gaussian gradient}  
	\nabla_{\theta} \log \pi_{\theta}(\boldsymbol{x};\boldsymbol{\mu}, \boldsymbol{\Sigma}) = \left[\boldsymbol{x}-\boldsymbol{\mu}\right]^T \boldsymbol{\Sigma}^{-1}   \left[\nabla_{\phi}\boldsymbol{\mu}\right]
\end{equation}
Taking $\boldsymbol{x}=\widetilde{\FTX}$ and $\boldsymbol{\mu} = \FTX({\boldsymbol{m}})$, as well as following the proof for Theorem \ref{thm:theorem1}, we have the theorem.
\end{proof}

 \begin{algorithm}[t]
   	\renewcommand{\algorithmiccomment}{\textbf{// }}
	\renewcommand{\algorithmicrequire}{\textbf{Input:}}
	\renewcommand{\algorithmicensure}{\textbf{Output:}}
   \caption{The training process of SemanticBC-SCAL.}
   \label{al:semanticbc_scal}
		\begin{algorithmic}[1]
                 \REQUIRE Batch size, learning rate $lr$, parallel samples $K$, input sequence ${\boldsymbol{m}}$, local iterations $\kappa$, no. of pre-training end epochs $E_p$, no. of end epochs $E_e$, no. of end batches $E_b$, number of \rxs ~$N$. 
			\ENSURE Encoder parameter $\theta$, decoder parameter $\phi_n$. \\
			\COMMENT{\textbf{Pre-training stage}} 
			\FOR{epoch=$1: E_p$} 
			\STATE Sample a batch of data to train parameters of encoder $\theta$ and decoder $\phi$ optimized by CE loss\cite{zhou2021semantic}.               
			\STATE Update $\theta$ and  $\phi$ to obtain the pre-trained parameters ${\theta}_{\text{pre}}$ and ${\phi}_{\text{pre}}$. \\  
			\ENDFOR\\
			\COMMENT{\textbf{RL-based alternate learning}}  
                \STATE $\theta  \leftarrow {\theta _{{\text{pre}}}}$, ${\phi _n} \leftarrow {\phi _{{\text{pre}}}}$
		\FOR{epoch=$E_p + 1: E_e$} 
                \FOR{$i=1: E_b$}
                \STATE Sample a batch of data
                \STATE 	\COMMENT{\textbf{Training Decoder}}
                \WHILE{ $(i\mod ~\kappa) ~ \ne  0$} 
               \FOR{$n=1: N$} 				
                \STATE Freeze $\theta$, ${\rm{R}}{{\rm{X}}_n}$ samples $K$ parallel trajectories, and updates ${\phi _n} \leftarrow {\phi _n} + lr \cdot {\nabla _{{\phi _n}}}J({\phi _n})$ as in (\ref{eq: decoder-gradient0}). 
     			      \ENDFOR 
                       \ENDWHILE
				   \STATE \COMMENT{\textbf{Training Encoder}}
				   \STATE For each ${\rm{R}}{{\rm{X}}_n}$, freeze $\phi_n$ and sample $K$ parallel trajectories.        
       \STATE Update $\theta  \leftarrow \theta  + lr \cdot {\nabla _{\theta }}J(\theta )$ by (\ref{eq: encoder-gradient1}).
           \ENDFOR
           \ENDFOR
                \STATE 	\COMMENT{\textbf{Finish training}} 
			\STATE Return $\theta$ and $\phi_n$. 		 	
		\end{algorithmic}
	\end{algorithm}	
 
\noindent \textbf{Remark:} On the basis of this stable and reasonable learning policy, our SCAL-based semantic BC system can adaptively encode/decode by introducing as few parameters and costs as possible. To obtain a stable and expected reward (semantic similarity score) in the training process, we adopt the aforementioned alternate learning mechanism at both $\rm{TX}$ and $\rm{RX}_n$ to implement a Monte-Carlo simulation. Specifically, in an \textit{update cycle}, the \rxs~ update $\kappa$ iterations and then \tx ~ updates once. Therefore, this process is implemented iteratively until both the encoder and decoders converge.  
With this updating mechanism through self-critical training, we provide a stable and precise gradient surrogate, in which we neither need an additional module to estimate the gradient, nor suffer from a poor estimation of expected reward. 

Finally, the detailed procedures for training SemanticBC-SCAL are summarized in \textbf{Algorithm \ref{al:semanticbc_scal}}.

\subsection{Convergence Analysis} 
\label{sec: Convergence Analysis} 
\begin{theorem}\label{thm:theorem3} 
For an independent decoder agent \rxn ~parameterized by $\phi_n$ for its policy updates, state value function converges to a fixed point in $V_{{\pi _{{\phi _n}}}}$.
\end{theorem}

\begin{proof}
According to the Bellman equation formulated in \cite{sutton2018reinforcement}, one has 
\begin{equation}
\label{eq: Bellman1}  
V_{{\pi _{{\phi _n}}}} = R_{{\pi _{{\phi _n}}}} + \lambda P_{{\pi _{{\phi _n}}}}V_{{\pi _{{\phi _n}}}}
\end{equation} 

We consider a complete metric space $\left\langle V_{{\pi _{{\phi _n}}}},d \right\rangle$, a mapping $F$: ${V_{{\pi _{{\phi _n}}}}} \to {V_{{\pi _{{\phi _n}}}}}$, where $d$ is a metric of nonempty set $V_{{\pi _{{\phi _n}}}}$ and satisfies $d: {V_{{\pi _{{\phi _n}}}}} \times {V_{{\pi _{{\phi _n}}}}} \to \mathbb{R}$. Hence, for any $\boldsymbol{x}$ in $V_{{\pi _{{\phi _n}}}}$,  the \eqref{eq: Bellman1} can be represented as $F({\boldsymbol{x}}) = R_{{\pi _{{\phi _n}}}} + \lambda P_{{\pi _{{\phi _n}}}}{\boldsymbol{x}}$. Meanwhile, for any $V_{{\pi _{{\phi _n}}}}^{(i)}$ and $V_{{\pi _{{\phi _n}}}}^{(j)}$ in $V_{{\pi _{{\phi _n}}}}$, $d$ is typically represented by the norm metric and here we choose infinity norm, \ie, $d(V_{{\pi _{{\phi _n}}}}^{(i)}, V_{{\pi _{{\phi _n}}}}^{(j)}) = \Vert V_{{\pi _{{\phi _n}}}}^{(i)}, V_{{\pi _{{\phi _n}}}}^{(j)}\Vert {_\infty }$. Then one can observe that   
\begin{align}
\label{eq: covergence-decoder}  
   &d( {F({V_{{\pi _{{\phi _n}}}}^{(i)}}),F({V_{{\pi _{{\phi _n}}}}^{(j)}}) } ) \nonumber \\
   = &\Vert(R_{{\pi _{{\phi _n}}}} + \lambda P_{{\pi _{{\phi _n}}}}{V_{{\pi _{{\phi _n}}}}^{(i)}}) - (R_{{\pi _{{\phi _n}}}} + \lambda P_{{\pi _{{\phi _n}}}}{V_{{\pi _{{\phi _n}}}}^{(j)}})\Vert{_\infty }  \nonumber \\ 
   =& \Vert \lambda P_{{\pi _{{\phi _n}}}}({V_{{\pi _{{\phi _n}}}}^{(i)}} - {V_{{\pi _{{\phi _n}}}}^{(j)}})\Vert _\infty \\
   \le & \Vert \lambda d({V_{{\pi _{{\phi _n}}}}^{(i)}} - {V_{{\pi _{{\phi _n}}}}^{(j)}})\Vert _\infty  \nonumber  \\
   = & \lambda d({V_{{\pi _{{\phi _n}}}}^{(i)}} - {V_{{\pi _{{\phi _n}}}}^{(j)}})  \nonumber 
\end{align}

Then according to the contracting mapping theorem (also known as the Banach Fixed-Point Theorem) \cite{agarwal2001fixed}, we can prove $F({\boldsymbol{x}}) = R_{{\pi _{{\phi _n}}}} + \lambda P_{{\pi _{{\phi _n}}}}{\boldsymbol{x}}$ has a fixed point in $V_{{\pi _{{\phi _n}}}}$, which satisfies $V_{{\pi _{{\phi _n}}}} = R_{{\pi _{{\phi _n}}}} + \lambda P_{{\pi _{{\phi _n}}}}V_{{\pi _{{\phi _n}}}}$. As such,  we conclude that, by continuously iterating from any state, the value function converges to a fixed point in $V_{{\pi _{{\phi _n}}}}$. Since our objective function (\ref{eq: objective function decoder}) aims to maximize the state value function from $s_n^{(0)}$, which satisfies the above conclusion.  
\end{proof} 

\begin{figure*}[ht] 
\begin{align}
   d( {F(V_{{\pi _\theta }}^{(m)},V_{{\pi _{{\phi _n}}}}^{(i)}),F(V_{{\pi _\theta }}^{(k)},V_{{\pi _{{\phi _n}}}}^{(j)})} )  
   = & \Vert F(V_{{\pi _\theta }}^{(m)},V_{{\pi _{{\phi _n}}}}^{(i)}) - F(V_{{\pi _\theta }}^{(m)},V_{{\pi _{{\phi _n}}}}^{(j)}) + F(V_{{\pi _\theta }}^{(m)},V_{{\pi _{{\phi _n}}}}^{(j)}) - F(V_{{\pi _\theta }}^{(k)},V_{{\pi _{{\phi _n}}}}^{(j)})\Vert_\infty  \nonumber \\ 
   \mathop  \le \limits^{(a)} & \Vert F(V_{{\pi _\theta }}^{(m)},V_{{\pi _{{\phi _n}}}}^{(i)}) - F(V_{{\pi _\theta }}^{(m)},V_{{\pi _{{\phi _n}}}}^{(j)})\Vert_\infty + ||F(V_{{\pi _\theta }}^{(m)},V_{{\pi _{{\phi _n}}}}^{(j)}) - F(V_{{\pi _\theta }}^{(k)},V_{{\pi _{{\phi _n}}}}^{(j)})\Vert_\infty \nonumber \\ 
   \mathop  = \limits^{(b)} & ||F(V_{{\pi _{{\phi _n}}}}^{(i)}) - F(V_{{\pi _{{\phi _n}}}}^{(j)})\Vert_\infty +\Vert F(V_{{\pi _\theta }}^{(m)}) - F(V_{{\pi _\theta }}^{(k)})\Vert_\infty  \label{eq: covergence-encoder}\\ 
    = & {\lambda _1}d(V_{{\pi _{{\phi _n}}}}^{(i)},V_{{\pi _{{\phi _n}}}}^{(j)})+ \Vert ({R_{{\pi _\theta }}} + {\lambda _2}{P_{{\pi _\theta }}}V_{\pi _\theta } ^{(m)}) - ({R_{{\pi _\theta }}} + {\lambda _2}{P_{{\pi _\theta }}}V_{\pi _\theta } ^{(k)})\Vert_\infty  \nonumber \\ 
  = & {\lambda _1}d(V_{{\pi _{{\phi _n}}}}^{(i)},V_{{\pi _{{\phi _n}}}}^{(j)}) +{\lambda _2}\Vert{P_{{\pi _\theta }}}(V_{{\pi _\theta }}^{(m)} - V_{{\pi _\theta }}^{(k)})\Vert_\infty 
   \le  {\lambda _1}d(V_{{\pi _{{\phi _n}}}}^{(i)},V_{{\pi _{{\phi _n}}}}^{(j)}) +{\lambda _2}d(V_{{\pi _\theta }}^{(m)},V_{{\pi _\theta }}^{(k)}) \nonumber
\end{align}
\hrulefill
\end{figure*}

\begin{corollary}\label{cor:mycorollary1}
Assuming that in each update cycle, \ie, the encoder updates once while the decoder undergoes $\kappa$ updates, which are assumed to be sufficient for all decoders converging to a fixed point, the encoder could ultimately reach the convergence as well. 
\end{corollary}
 
\begin{proof}
First, under Theorem \ref{thm:theorem3}, it suffices to show that the decoder can be converged under policy ${{\pi _{{\phi _n}}}}$ guided by the Bellman equation in (\ref{eq: Bellman1}). For the encoder learning, Bellman equation is supervised by Gaussian policy $\pi_{\theta}$, that is $V_{\pi_{\theta}} = R_{\pi_{\theta}} + \lambda P_{\pi_{\theta}}V_{\pi_{\theta}}$. We define $F({\boldsymbol{x}}) = R_{{\pi _{{\theta}}}} + \lambda P_{{\pi _{{\theta}}}}{\boldsymbol{x}}$, then for any $V_{{\pi _{{\theta}}}}^{(m)}$ and $V_{{\pi _{{\theta}}}}^{(k)}$ in set $V_{{\pi _{{\theta}}}}$, $V_{{\pi _{{\phi_n}}}}^{(i)}$ and $V_{{\pi _{{\phi_n}}}}^{(j)}$ in set $V_{{\pi _{{\phi_n}}}}$, one has (\ref{eq: covergence-encoder}). Specifically, inequality (a) comes from the 
triangle inequality, and for equality (b), the first term is due to policy $\pi_{\phi_n}$ making the decision, while similarly, the second term stems from the pivotal role played by policy $\pi_\theta$. Thus, we get the corollary. 
\end{proof}

\noindent \textbf{Remark:} As implied by Corollary \ref{cor:mycorollary1}, we optimize the encoder and decoders asynchronously with two different iteration counts, \ie, two coupled iterations that the decoders at \rxs ~ update more frequently while the encoder at \tx ~updates at a lower speed. Convergence of these interleaved iterations can be ensured by assuming that the update frequency of the encoder is considerably smaller (\ie, $\kappa$ smaller) than decoders, which allows decoders to converge first while being perturbed by the slower encoder\cite{holzleitner2021convergence}.

 \section{Performance Evaluation}
  \label{sec: Performance Evaluation}
 
 In this section, we compare the proposed SemanticBC-SCAL with the traditional reliable communication scheme and typical semantic communication algorithms under both AWGN and Rayleigh fading channels respectively. In addition, we also give an ablation study to comprehensively evaluate and analyze the performance.  
 \subsection{Simulation Settings}   
 
  \label{sec: Simulation Settings} 
We adopt the standard English version of the proceedings of the European Parliament \cite{koehn2005europarl} as the dataset, which consists of more than $2.0$ million sentences and $53$ million words. In our experiment, the dataset is pre-processed into sentences with the length of $4\sim30$ words, those shorter than 30 words being padded with zeros and the number of units $D$ in the hidden layer behind each word equals $128$. Furthermore, each word is converted to $30$ bits (\ie, $B=30$) by the quantization layer for further transmission. The pre-proposed dataset contains a total of $1,136,816$ sentences and $32,478$ words (dictionary size), which are then divided into training and testing with a ratio of $4:1$. Furthermore, the backbone of semantic encoder/decoders is composed of UT \cite{dehghani2018universal} while the channel encoder and decoder are based on dense layers. The detailed DNN structure and hyperparameters are summarized in Table \ref{structure}.

\begin{table}[t]
		\centering
		\caption{Network structure and hyperparameters used in SemanitcBC-SCAL}  
		\label{structure}		 
		\def\arraystretch{1.0}  
    \begin{tabular}{m{0.8cm} m{1.0cm}  m{1.9cm}  m{1.8cm}  m{1.2cm}} 	
			\hline
			&Modules &Layer&Dimensions & Activation\\ 
			\hline   
			\multirow{8}{0.8cm}{${\rm{TX}}$}&Input&${\boldsymbol{m}}$&30&$\backslash$	\\ 
                \hline 
                &$S{C_{{\rm{en}}}}( \cdot )$& Embedding \newline UT Encoder \newline Dense layer \newline Gaussian policy& $(30 \times 128)$\newline $(30\times128)$\newline $(30\times16)$ \newline $(30\times16)$& $\backslash$ \newline$\backslash$\newline ReLU \newline $\backslash$ \\ 
                & $Q$ & Dense layer \newline quantization &  $(30\times30)$\newline$(30\times30)$ & ReLU\newline $\backslash$ \\ 
                \hline 
                Channel& AWGN & $\backslash$&$\backslash$&$\backslash$\\ 
                \hline 
                \multirow{5}{0.8cm}{\textbf{$3\times{{\rm{R}}{{\rm{X}}_n}}$}}& $Q_n^{ - 1}( \cdot )$ & Dense layer&  $(30\times16)$ & ReLU\\ 
                &$S{C_{n,{\rm{de}}}}( \cdot )$ & Dense layer \newline UT Decoder \newline Dense layer  \newline Softmax & $ (30\times128)$\newline $(30\times128)$\newline $(30\times32,478)$ \newline $(32,478\times{{\hat T}_n})$ & ReLU\newline  ReLU\newline Dictionary \newline  Softmax\\	 
			\hline 
		\end{tabular}
\end{table}

As for the wireless channel, we consider the commonly used AWGN and Rayleigh fading channel as discussed in Section \ref{sec: system_model}. As a comparison, we utilize the following baselines. 
\begin{enumerate}
    \item Huffman-RS: A traditional communication system that source coding and channel coding are implemented as Huffman coding and Reed-Solomon (RS) \cite{reed1960polynomial} coding, respectively. Meanwhile, the coding length with RS is set to $(30\times 42)$ in $ 64$-QAM.  
    \item CE baseline: It refers to a degenerative point-to-point model that shares the same network structure as our proposed model. Put differently, a special case arises when SemanticBC-SCAL is optimized by the CE function.  
    \item LSTM: It optimizes the entire semantic system on top of RL but its backbone is built around LSTM (Long Short-Term Memory) networks \cite{lu2021reinforcement}.  
    \item DeepSC: It belongs to the pioneering work in semantic communication \cite{xie2021deep}, and leverages the transformer architecture optimized by CE function and mutual information. 
\end{enumerate}

\begin{figure*} [hb]
\subfigure[AWGN]{
  \begin{minipage}[t]{1.0\linewidth} 
    \centering 
    \includegraphics[width=180mm]{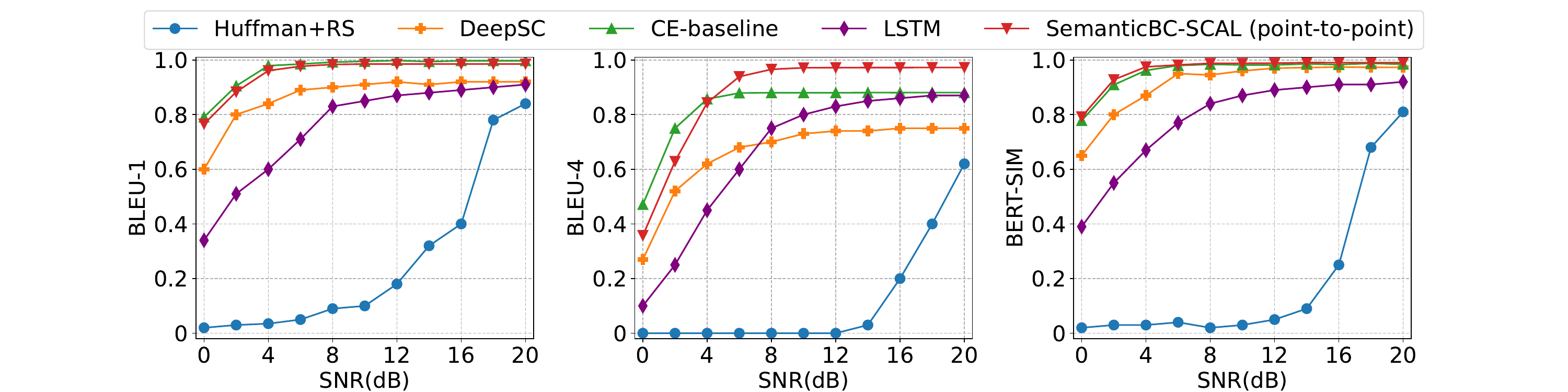}
  \end{minipage}}
\subfigure[Rayleigh Fading]{
  \begin{minipage}[t]{1.0\linewidth} 
    \centering 
    \includegraphics[width=180mm]{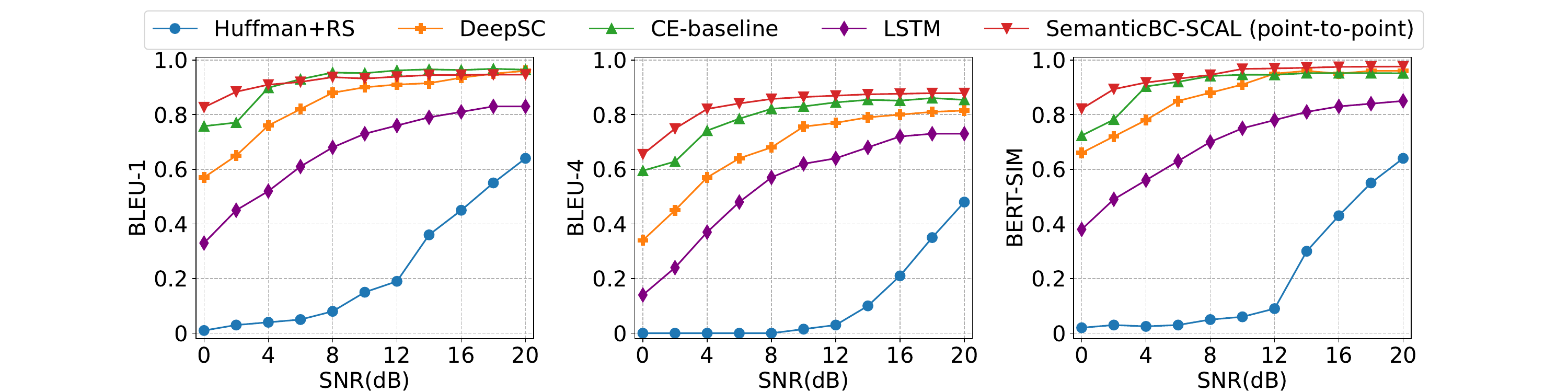}   
  \end{minipage} }  
  \caption{BLEU score and BERT-SIM versus SNR in AWGN and Rayleigh fading channel for the point-to-point case.}     
  \label{1-1-awgn-fading-bleu}  
\end{figure*} 

Moreover, for SemanticBC-SCAL, the pre-training process is implemented by utilizing the CE function to accelerate the training. 
 Afterwards, all pre-trained \rxs ~undergo the described training procedure in SemanticBC-SCAL until convergence. Besides, in order to characterize the inherently noisy channel environment, we have implemented two types of SNR simulation methods. Firstly, we set deterministic channel SNR \ie, $\mu _{\rm SNR}$ for each receiver to indicate practical channel conditions. Secondly, in broadcast cases, we add different noise variances $\delta _{\rm SNR}^2$ of SNR for different \rxs ~to reflect their heterogeneous estimation capabilities and adaptability. Specially, for both AWGN and Rayleigh fading channels, we set different mean values ${\mu _{\rm SNR}}$ and noise variance $\delta _{\rm SNR}^2$ during the training stage respectively, including $ {\mu _{{\rm{SNR},1}}} = 6$ dB, $\delta _{\rm {SNR},1} = 1$ dB; $\mu _{\rm{SNR},2} = 10$ dB, $\delta _{\rm{SNR},2} = 1$ dB; $\mu _{\rm{SNR},3} = 10$ dB, $\delta _{\rm{SNR},3} = 2$ dB. While for SemanticBC-SCAL in the point-to-point scenario, we set the SNR of the training process by default with $ \mu _{\rm{SNR}} = 10$ dB and $\delta _{\rm{SNR}} = 1$ dB. The related parameters are listed in Table \ref{default parameters}.  

\begin{table}[t]
		\centering
		\caption{The default parameters in SemanitcBC-SCAL}    
		\label{default parameters}		  
		\def\arraystretch{1.0}  
    \begin{tabular}{m{1cm} m{4cm}m{2cm} } 	
			\hline
			Parameters &Description&Value\\ 
                \hline
                $\backslash$&Batch size& $64$ \\ 
                $lr$&Learning rate & $1e-5$\\
                $K$&Number of parallel samples &$5$\\
                $\kappa$& Local iterations for decoders in each \textit{update cycle}&$1000$\\ 
                $E_p$ &No. of pre-training end epochs& $50$\\
                $E_e$&No. of end epochs&$180$\\               
                \hline 
		\end{tabular}
        \vspace{-0.35cm}
\end{table}

To be consistent with existing works on SemCom, we adopt BLEU scores \cite{papineni2002bleu}, BERT-SIM \cite{devlin2018bert}, and WAR to evaluate the system's performance. In particular, BLEU scores count the similarity of $n$-gram phrases from the recovered text and the referred ground truth, which can be typically denoted as BLEU-1, BLEU-2, BLEU-3, and BLEU-4 respectively. BERT-SIM calculates the cosine similarity of BERT embeddings after fine-tuning. In some sense, WAR shares some similarity with $1$-gram. Nevertheless, BLEU scores \cite{papineni2002bleu} consider the length of  sentences by introducing a brevity penalty. In other words, if the length of decoded sentences doesn't perfectly match the reference sentence, the computed BLEU score will be penalized.
In terms of semantic similarity metric $\Theta$ for reward training, we choose the BLEU score by combining 1-gram, 2-gram, 3-gram, and 4-gram settings with equal weights $0.25$, so as to provide semantic level supervision.

 \subsection{Numerical Results and Analysis}

 \subsubsection{Performance in Point-to-Point SemCom} 

We start with the performance comparison between the SemanticBC-SCAL with other baselines in the point-to-point case. In particular, 
we provide the performance in terms of BLEU scores and BERT-SIM under AWGN and Rayleigh fading channel in Fig. \ref{1-1-awgn-fading-bleu}.  
As shown in Fig. \ref{1-1-awgn-fading-bleu}(a),    
our proposed SemanticBC-SCAL outperforms CE baseline and LSTM, especially in terms of BLEU-4. Since a larger $n$-gram phrase carries more contextual information, the result demonstrates our approach is capable of preserving semantics and prone to provide a semantic level transmission. Furthermore, our approach also outperforms baselines in terms of BERT-SIM, allowing for accommodating some synonyms. Meanwhile, as illustrated in Fig. \ref{1-1-awgn-fading-bleu}(b), attributed to the volatile channel gain, the Rayleigh fading channel imposes a stronger detrimental impact compared to the AWGN channel on both the traditional method and SemCom models, leading to a more pronounced attenuation during the signal transmission.   

Furthermore, our method exhibits significant superiority over the traditional Huffman-RS method, and stands out among other DNN-based SemCom approaches, especially in low SNRs. Most impressively, our approach obtains a decoding accuracy that closely approximates $100\%$ at around ${\rm SNR}=8$ dB in AWGN channels. 
This improved performance can be primarily attributed to the adaptive transformer layer in the UT, which has the ability to reconstruct partial semantic information in a poor channel environment by leveraging shared background knowledge, thereby mitigating the distortion introduced by channel noise. Meanwhile, by adopting semantic similarity scores as rewards, our approach fosters an environment conducive to the system's learning and understanding of semantic representations at a semantic level, which addresses the $Q2$.

\begin{figure*} [ht]
\subfigure[AWGN]{
  \begin{minipage}[t]{1.0\linewidth} 
    \centering 
    \includegraphics[width=180mm]{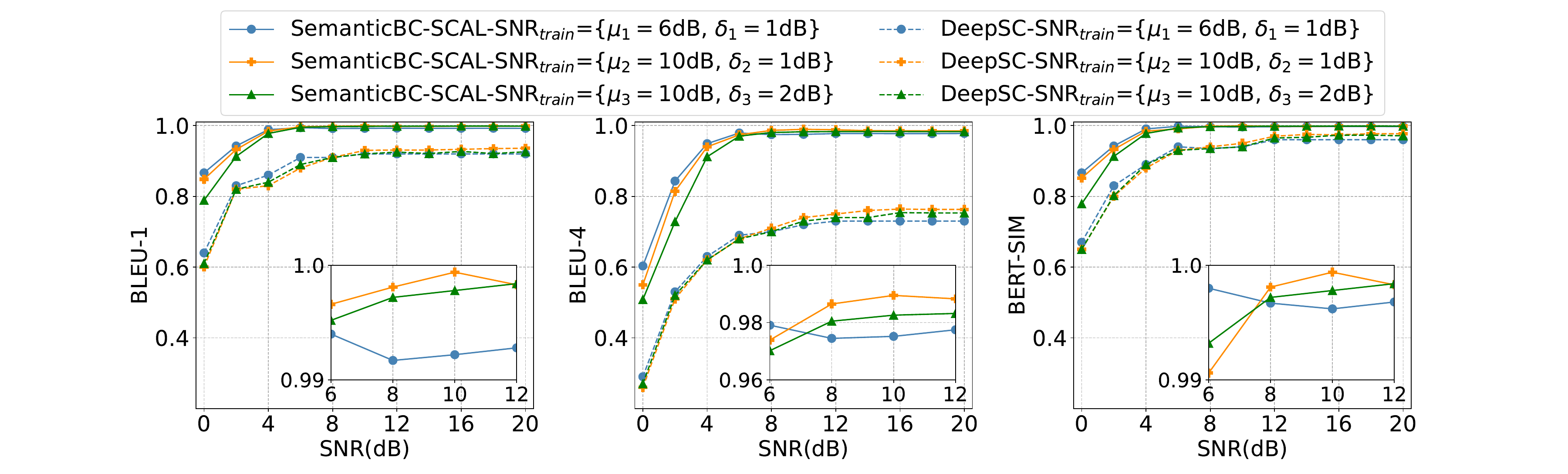}
  \end{minipage}}
\subfigure[Rayleigh Fading]{
  \begin{minipage}[t]{1.0\linewidth} 
    \centering 
    \includegraphics[width=180mm]{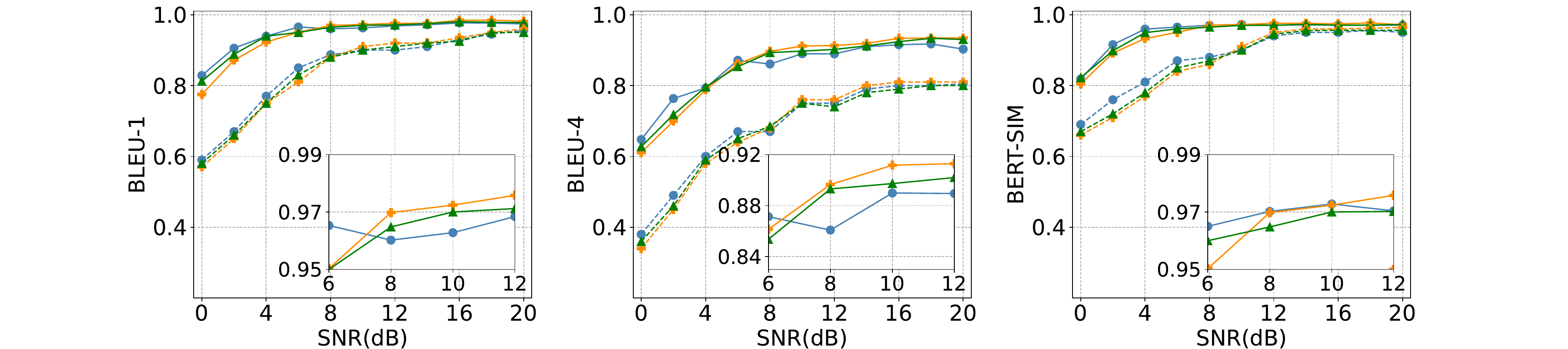}   
  \end{minipage} }  
  \caption{Comparison of Semantic BC schemes in terms of BLEU-1 score, BLEU-4 score, and BERT-SIM under both AWGN and Rayleigh fading channel.}     
  \label{1-3-awgn-fading} 
\end{figure*}

Notably, even if our approach has a lower BLEU-1 compared to the CE-baseline in Fig. \ref{1-1-awgn-fading-bleu}, the existence of synonyms in the recovered sentences actually leads to a slightly higher sentence similarity in terms of BLEU-4. More importantly, although BERT-SIM is not directly adopted as a reward function, the semantic level supervision under BLEU score also leads to the superiority of SemanticBC-SCAL in the point-to-point case and manifests its effectiveness.

 \subsubsection{Performance in Semantic BC} 
\label{sec: Performance in Point-to-Point SemCom} 

For semantic BC,  
in addition to the inherent channel noise as discussed in (\ref{eq:decoding}), we also consider the estimation variance related to SNRs to better adapt to the varying conditions. Consequently, during the training process, we employ lower SNR to acclimate to these demanding channel conditions, \ie, exemplified by scenarios, such as the AWGN and Raileigh channel under $ {\mu _{{\rm{SNR},1}}} = 6$ dB, $\delta _{\rm {SNR},1} = 1$ dB. 
Conversely, in favorable channel environment, we frequently employ higher SNR settings during training to guarantee robust performance across a substantial portion of the channel environment while maintaining a competitive level of resilience. Fig. \ref{1-3-awgn-fading} visualizes the corresponding BLEU-1, BLEU-4 and BERT-SIM in Semantic BC with three different receivers under both AWGN and Rayleigh channels. 
 
 \begin{table}[htb]
		\centering
		\caption{Comparisons on the varying number of RXs of SemanticBC-SCAL under AWGN channel, wherein all BLEU scores and WAR denote the averages with SNR ranging from $0\sim20$ dB.} 
		\label{varying users performance}

		\def\arraystretch{1.0} 
		
		\begin{tabular}{p{1.2cm}<{\centering}p{1cm}<{\centering}p{1cm}<{\centering}p{1cm}<{\centering}p{1cm}<{\centering}p{0.8cm}<{\centering}}			
			\hline
			Number of RXs&BLEU-1&BLEU-2&BLEU-3&BLEU-4&WAR\\ 
			\hline
			1&0.9528&0.9262&0.9004&0.8697&0.9272 \\ 
			3&0.9760&0.9583&0.9442&0.9209&0.9614  \\ 
			5&0.9624&0.9338&0.9120&0.8834&0.9466 \\ 
                7&0.9517&0.9288&0.9018&0.8811&0.9358 \\ 
			\hline
		\end{tabular}
\end{table}	

From Fig. \ref{1-3-awgn-fading}, it can be observed that our SemanticBC-SCAL demonstrates adaptability to the varying SNR, particularly excelling in low SNR regime. Notably, when the test SNR ranges from $0$ to $6$ dB, the \rxn ~trained under low SNR (i.e, $ {\mu _{{\rm{SNR},1}}} = 6$ dB, $\delta _{\rm {SNR},1} = 1$ dB) yields superior results than others for both AWGN and fading channels, while DeepSC in BC scenarios follows a similar trend but has an inferior performance.  
To elaborate, when the mean value of training SNR is configured at $ {\mu _{{\rm{SNR},1}}} =  10$ dB, 
for SemanticBC-SCAL in the AWGN channel, the training SNR with a narrower variance (i.e., $\delta _{\rm {SNR},1} = 1$ dB) outperforms that with a wider variance (i.e., $\delta _{\rm {SNR},1} = 2$ dB) in a low SNR regime, while in Rayleigh fading channels the opposite is true, thanks to its generalization capabilities. 
 Meanwhile, in the higher SNR regime, the performance of the channel trained under low SNR exhibits insubstantial improvement and the channel with lower variance estimation performs better.

 \begin{figure*} [ht]
\hspace{-3mm}
\subfigure[]{
  \begin{minipage}[t]{0.5\linewidth} 
    \centering 
    \includegraphics[width=70mm]{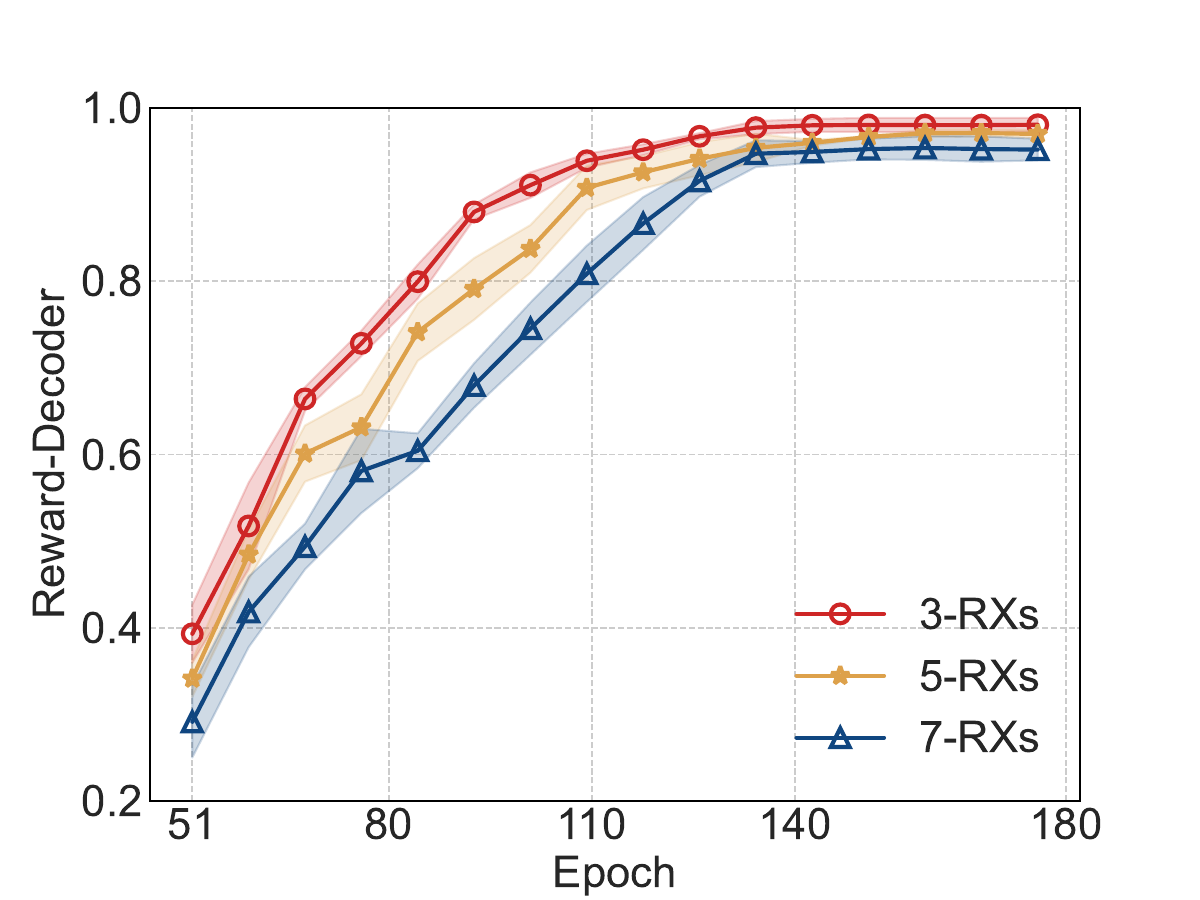} 
  \end{minipage}} \hspace{-1mm}
\subfigure[]{
  \begin{minipage}[t]{0.5\linewidth} 
    \centering 
    \includegraphics[width=70mm]{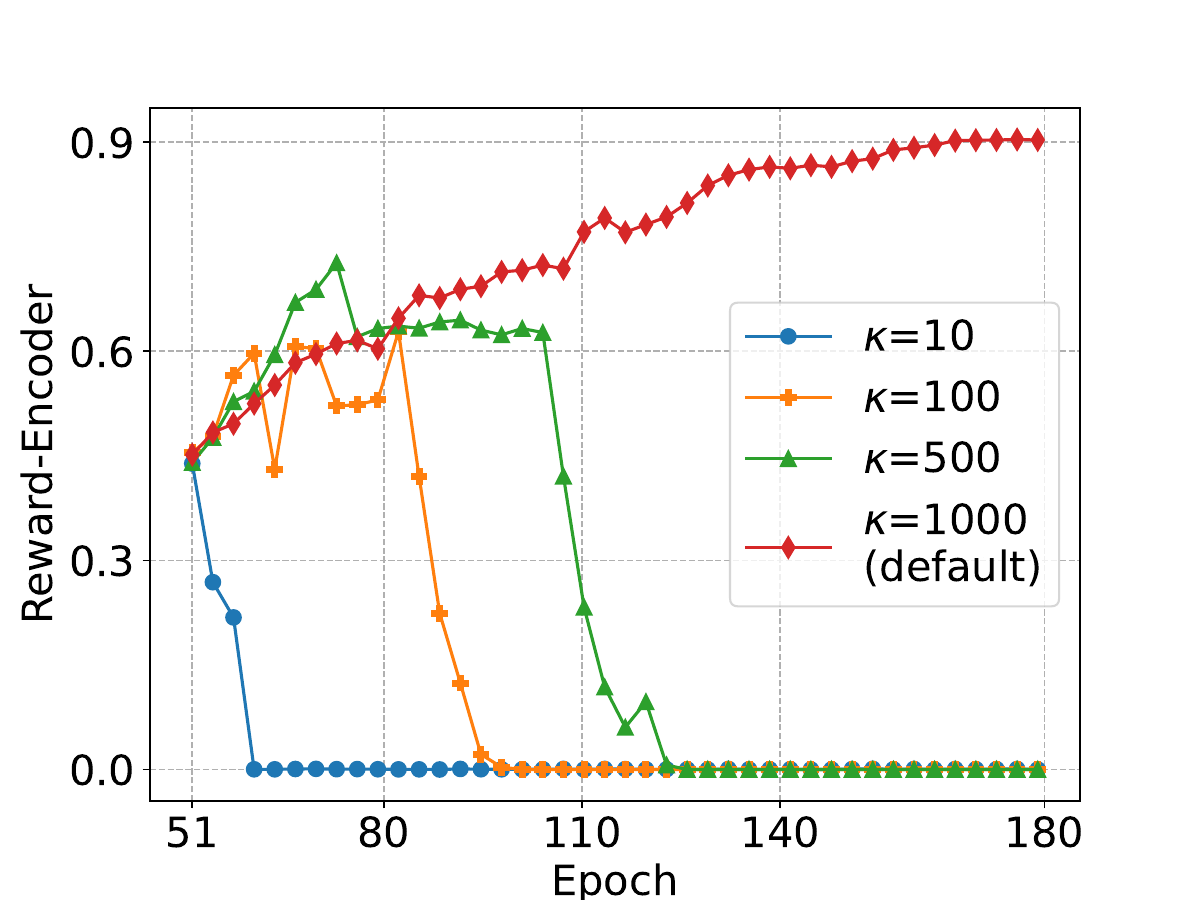} 
  \end{minipage} }
  \caption{Convergence of SemanticBC-SCAL, wherein the reward curves are the average BLEU score for all broadcast receivers. (a) The training reward of the decoder sides with varying numbers of \rxs. (b) The training reward of the encoder side with varying local iterations $\kappa$ under three \rxs.  
 }
  \label{decoder-encoder-reward}
\end{figure*}

Additionally, we also validate the scalability of our semantic BC system in scenarios with a varying number of \rxs, as depicted in Table \ref{varying users performance}. It is noteworthy that as the number of RX increases, there is a modest decrease in semantic decoding proficiency, but the overall BLEU scores and WAR consistently maintain a high level. Furthermore, when compared to the point-to-point system, our SemanticBC-SCAL demonstrates a notable advantage in utilizing semantics within the broadcast scenario. This advantage can be attributed to the localized adaptive learning and comprehensive integration of all decoders' outcomes. This indicates that our SemanticBC-SCAL can be effectively extended to semantic BC systems with varying numbers of \rxs, thereby addressing the $Q1$.

 \subsubsection{Convergence of SemanticBC-SCAL}   
 \label{sec: Convergence of Semantic-SCAL} 
In order to testify whether our alternate learning mechanism guarantees convergence, we have depicted the evolution of the reward for the encoder and decoders over epochs in Fig. \ref{decoder-encoder-reward}. From Fig. \ref{decoder-encoder-reward}(a), we can observe that consistent with Theorem \ref{thm:theorem3}, our alternate learning mechanism ultimately converges with a steady learning and low variance (among different receivers) for the different number of \rxs, which promises an effective learning approach for the large-scale knowledge base. 

Furthermore, as highlighted in the ``Remark'' accompanying Corollary $1$, the convergence of these interleaved iterations between one encoder and multiple decoders can be ensured by certain locally updated iterations $\kappa$. Consequently, we delve into exploring the influence of $\kappa$ on the convergence of the encoder, as illustrated in Fig. \ref{decoder-encoder-reward}(b). Based on the observations from Fig. \ref{decoder-encoder-reward}(b), it becomes apparent that when $\kappa$ is relatively small, indicating a low number of local iterations for decoders and frequent updates for the encoder per epoch. In this sense, the encoder fails to learn effectively even after a few training epochs. As $\kappa$ increases, the encoder can sustain longer training epochs but still struggles to converge, until $\kappa=1000$. Only at this point, it achieves continuous and stable learning, ultimately reaching the convergence, which aligns with those presented in Corollary $1$. These empirical findings clearly address the $Q1$ and $Q3$.

  \begin{figure}[t]
   \centering
   \setlength{\abovecaptionskip}{0.3cm}
		\includegraphics[width=0.85\linewidth]{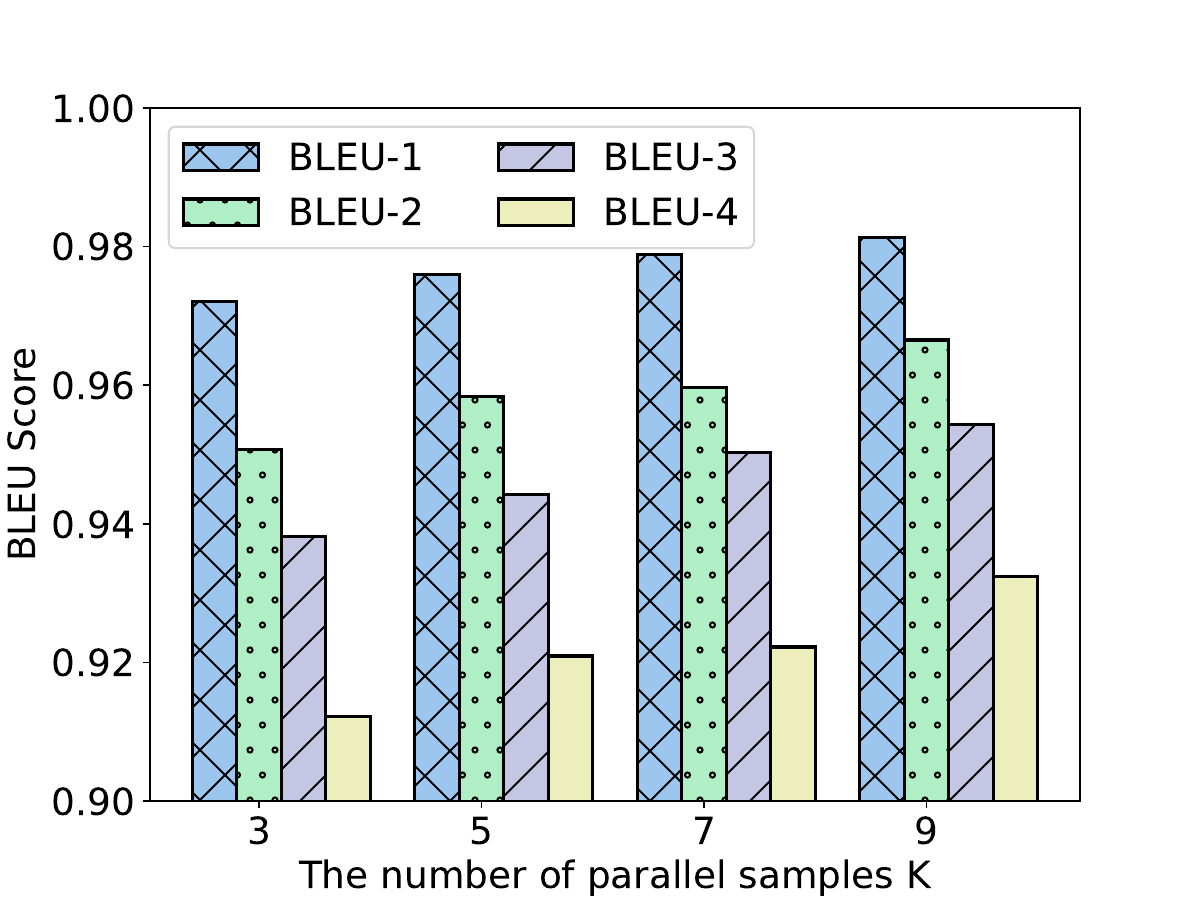}\\     
		\caption{The impact of the number of parallel samples $K$ in the SemanticBC-SCAL system.} 
		\label{sample-K} 
\end{figure}
 \subsubsection{Ablation Study}
 
The ablation study for comparing different numbers of parallel samples $K$ in (\ref{eq: decoder-gradient0}) and (\ref{eq: encoder-gradient1}) is given in Fig. \ref{sample-K}. According to (\ref{eq: decoder-gradient0}) and (\ref{eq: encoder-gradient1}), as $K$ increases, the mean reward from $K$ rollouts of state-action value can be estimated more precisely. Consistently, as depicted in Fig. \ref{sample-K}, a larger $K$ obtains a higher semantic similarity, which is consistent with the referred \cite{gao2019self}.

\section{Conclusions}
\label{sec: Conclusions}

In this paper, in order to cope with scalability, compatibility, and convergence issues in semantic broadcast communications, we have proposed a semantic broadcast communication framework optimized by an RL-based self-critical alternate learning, called SemanticBC-SCAL, which provides semantic level supervision and improves the robustness under varying numbers of \rxs. In particular, we have regarded semantic similarity as a reward and formulated the optimization process as an RL problem. Furthermore, we have adopted self-critical learning to obtain a simple and efficient gradient estimation, especially suitable for complex SemCom systems with large-scale samples. To accommodate varying numbers of \rxs~ and address non-differentiable channels, we have adopted a cost-effective alternate training mechanism to asynchronously learn the encoder and multiple decoders with different learning iterations instead of directly backpropagating gradients through the channel. Moreover, we have also delved into the convergence analysis of SemanticBC-SCAL and provided conditions for the convergence. Extensive simulation results have demonstrated that our proposed semantic broadcast communication system can achieve high semantic accuracy with different numbers of \rxs, and has exhibited strong scalability and robustness for different channel environments.


	\bibliographystyle{IEEEtran}
	\bibliography{Semantic_BC}

	
	
	
	
	
	\end{document}